\documentclass{article}

\usepackage[preprint]{neurips_2024}

\usepackage[utf8]{inputenc} % allow utf-8 input
\usepackage[T1]{fontenc}    % use 8-bit T1 fonts
\usepackage{hyperref}       % hyperlinks
\usepackage{url}            % simple URL typesetting
\usepackage{booktabs}       % professional-quality tables
\usepackage{amsfonts}       % blackboard math symbols
\usepackage{nicefrac}       % compact symbols for 1/2, etc.
\usepackage{microtype}      % microtypography
\usepackage{xcolor}         % colors

\usepackage{amsmath,amssymb,amsthm}
\usepackage{mathtools}
\usepackage{algorithm,algorithmic,setspace}
\usepackage{esvect}
\usepackage{bm}
\usepackage{subcaption}

\hypersetup{
    colorlinks = true,
    linkcolor = blue
    }
\usepackage[
        nameinlink,
    ]{cleveref}

\theoremstyle{plain}
\newtheorem{theorem}{Theorem}[section]

\newtheorem{lemma}[theorem]{Lemma}

\theoremstyle{definition}

\theoremstyle{remark}

\newcommand\distribeq{\mathrel{\overset{\makebox[0pt]{\mbox{\normalfont\tiny\sffamily d}}}{=}}}
\newcommand\distribapprox{\mathrel{\overset{\makebox[0pt]{\mbox{\normalfont\tiny\sffamily d}}}{\approx}}}

\providecommand{\keywords}[1]
{
  \small	
  \textbf{\textit{Keywords---}} #1
}

\newcommand{\bfzero}{\mathbf 0}

\newcommand{\bfa}{\mathbf a}
\newcommand{\bfb}{\mathbf b}

\newcommand{\bff}{\mathbf f}

\newcommand{\bfI}{\mathbf I}
\newcommand{\bfK}{\mathbf K}
\newcommand{\bfL}{\mathbf L}

\newcommand{\bfr}{\mathbf r}
\newcommand{\bfp}{\mathbf p}
\newcommand{\bfP}{\mathbf P}

\newcommand{\bfS}{\mathbf S}

\newcommand{\bfU}{\mathbf U}
\newcommand{\bfu}{\mathbf u}
\newcommand{\bfv}{\mathbf v}
\newcommand{\bfX}{\mathbf X}
\newcommand{\bfx}{\mathbf x}
\newcommand{\bfy}{\mathbf y}
\newcommand{\bfz}{\mathbf z}
\newcommand{\bfZ}{\mathbf Z}

\newcommand{\bfbeta}{\boldsymbol{\beta}}

\newcommand{\bfmu}{\boldsymbol{\mu}}

\newcommand{\bfSigma}{\boldsymbol{\Sigma}}
\newcommand{\bfepsilon}{\boldsymbol{\epsilon}}
\newcommand{\calA}{\mathcal A}

\newcommand{\calGP}{\mathcal{GP}}

\newcommand{\calN}{\mathcal N}
\newcommand{\calO}{\mathcal O}
\newcommand{\calP}{\mathcal P}
\newcommand{\calU}{\mathcal U}

\newcommand{\calX}{\mathcal X}
\newcommand{\bbE}{\mathbb E}
\newcommand{\bbG}{\mathbb G}
\newcommand{\bbN}{\mathbb N}
\newcommand{\bbP}{\mathbb P}
\newcommand{\bbR}{\mathbb R}

\newcommand{\hatf}{\hat{f}}

\title{Gaussian Processes Sampling with Sparse Grids\\ under Additive Schwarz Preconditioner}

\author{%
Haoyuan Chen$^{1}$\thanks{Submitted to the editors on January 31, 2024.}
\quad 
    Rui Tuo$^{1}$\\
$^1$Department of Industrial and Systems Engineering, Texas A\&M University\\ College Station, TX 77843 \\
\texttt{chenhaoyuan2018@tamu.edu}\\
\texttt{ruituo@tamu.edu}
}

\begin{document}

\maketitle

\begin{abstract}
  Gaussian processes (GPs) are widely used in non-parametric Bayesian modeling, and play an important role in various statistical and machine learning applications. In a variety tasks of uncertainty quantification, generating random sample paths of GPs is of interest. As GP sampling requires generating high-dimensional Gaussian random vectors, it is computationally challenging if a direct method, such as the Cholesky decomposition, is used. In this paper, we propose a scalable algorithm for sampling random realizations of the prior and posterior of GP models. The proposed algorithm leverages inducing points approximation with sparse grids, as well as additive Schwarz preconditioners, which reduce computational complexity, and ensure fast convergence. We demonstrate the efficacy and accuracy of the proposed method through a series of experiments and comparisons with other recent works.
\end{abstract}


%%%%%%%%%%%%%%%%%%%%%%%%%%%%%%%%%%%%%%%%%%%%%%%%%%%%%%%%%%%%

\section{Introduction}
\textit{Gaussian processes} (GPs) hold significant importance in various fields, particularly in statistical and machine learning applications \citep{rasmussen2003gaussian,santner2003design,cressie2015statistics}, due to their unique properties and capabilities. As probabilistic models, GPs are highly flexible and able to adapt to a wide range of data patterns, including those with complex, nonlinear relationships. GPs have been successfully applied in numerous domains including regression \citep{o1978curve,bishop1995neural,mackay2003information}, classification \citep{kuss2005assessing,nickisch2008approximations,hensman2015scalable}, time series analysis \citep{girard2002gaussian,roberts2013gaussian}, spatial data analysis \citep{banerjee2008gaussian,nychka2015multiresolution}, Bayesian networks \citep{neal2012bayesian}, optimization \citep{srinivas2009gaussian}, and more. This wide applicability showcases their versatility in handling different types of problems. %GP modeling proceeds by imposing a GP as the prior of an underlying continuous function, which provides a flexible nonparametric framework for prediction and inference problems. %It's a flexible framework for simulation and inference in machine learning area. Unfortunately, Gaussian process methodology suffers from a lot of computational issues when coming across large datasets or high dimensional space.
As the size of the dataset increases,  GP inference suffers from computational issues of the operations on the covariance matrix such as inversion and determinant calculation.
% When the sample size is large, the basic framework for GP regression suffers from the computational challenge of inverting large covariance matrices. A lot of work has been done to address this issue.%To solve the computational challenges, a huge amount of approaches were proposed. 
These computational challenges have motivated the development of various approximation methods to make GPs more scalable. Recent advances in scalable GP regression include Nystr\"om approximation  \citep{quinonero2005unifying,titsias2009variational,hensman2013gaussian}, random Fourier features \citep{rahimi2007random,hensman2018variational}, local approximation \citep{gramacy2015local,cole2021locally}, structured kernel interpolation \citep{wilson2015kernel,stanton2021kernel}, state-space formulation \citep{hartikainen2010kalman, grigorievskiy2017parallelizable,nickisch2018state}, Krylov methods \citep{gardner2018gpytorch,wang2019exact}, Vecchia approximation \citep{katzfuss2021general,katzfuss2022scaled}, sparse representation \citep{ding2021sparse,chen2022kernel,chen2022scalable}, etc.

%{\color{red}[Significance]}
The focus of this article is on generating samples from GPs, which includes sampling from both prior and posterior processes. GP sampling is a valuable tool in statistical modeling and machine learning, and is used in various applications, such as in Bayesian optimization \citep{snoek2012practical, frazier2018bayesian, frazier2018tutorial}, in computer model calibration \citep{kennedy2001bayesian,tuo2015efficient,plumlee2017bayesian,xie2021bayesian}, in reinforcement learning \citep{kuss2003gaussian,engel2005reinforcement,grande2014sample}, and in uncertainty quantification \citep{murray2004transformations,marzouk2009dimensionality,teckentrup2020convergence}. Practically, sampling from GPs involves generating random vectors from a multivariate normal distribution after discretizing the input space. However, sampling from GPs poses computational challenges particularly when dealing with large datasets due to the need to factorize the large covariance matrices, which is an $\calO(n^3)$ operation for $n$ inputs. There has been much research on scalable GP regression, yet sampling methodologies remain relatively unexplored.
Scalable sampling techniques for GPs are relatively rare in the existing literature. A notable recent development in the field is the work of \citet{wilson2020efficiently}, who introduced an efficient sampling method known as decoupled sampling. This approach utilizes Matheron's rule and combines the inducing points approximation with random Fourier features. Furthermore, \citet{wilson2021pathwise} expanded this concept to pathwise conditioning based on Matheron's update. This method, which requires only sampling from GP priors, has proven to be a powerful tool for both understanding and implementing GPs. Inspired by this work, \citet{maddox2021bayesian} applied Matheron's rule to multi-task GPs for Bayesian optimization. Additionally, \citet{nguyen2021gaussian} applied such techniques to tighten the variance and enhance GP posterior sampling. \citet{wenger2022posterior} introduced IterGP, a novel method that ensures consistent estimation of combined uncertainties and facilitates efficient GP sampling by computing a low-rank approximation of the covariance matrix. \citet{lin2023sampling} developed a method for generating GP posterior samples using stochastic gradient descent, targeting the complementary problem of approximating these samples in $\calO(n)$ time. It is important to note, however, that none of these methods have explored sampling within a structured design, an area that possesses many inherent properties.

%{\color{red}[In this paper, we propose ...]}
In this study, we propose a novel algorithms for efficient sampling from GP priors and GP regression posteriors by utilizing inducing point approximations on sparse grids. We show that the proposed method can reduce the time complexity to $\calO\big( (\eta + n)n_{sg} )$ for input data of size $n$ and inducing points on sparse grids of level $\eta$ and size $n_{sg}$. This results in a linear-time sampling process, particularly effective when the size of the sparse grids is maintained at a reasonably low level. Specifically, our work makes the following contributions:
\begin{itemize}
    \item We propose a novel and scalable sampling approach for Gaussian Processes using product kernels on multi-dimensional points. This approach leverages the inducing point approximation on sparse grids, and notably, its computational time scales \textit{linearly} with the product of the sizes of the inducing points and the sparse grids. 
    \item Our method provides a $2$-Wasserstein distance error bound with a convergence rate of $\calO( n_{sg}^{-\nu} (\log n_{sg})^{(\nu+2)(d-1)+d+1} )$ for $d$-dimensional sparse grids of size $n_{sg}$ and product Mat\'ern kernel of smoothness $\nu$. %We propose a scalable sampler for the posterior Gaussian processes. The proposed method is exact for one-dimensional data or multi-dimensional data on full grid designs.
    %based on one-dimensional data or multi-dimensional data on full grid designs.
    \item We develop a two-level additive Schwarz preconditioner to accelerate the convergence of conjugate gradient during posterior sampling from GPs.
    \item We illustrate the effectiveness of our algorithm in real-world scenarios by applying it to areas such as Bayesian optimization and dynamical system problems.
\end{itemize}

% The outline is not required, but we show an example here.
The structure of the paper is as follows: Background information is presented in \Cref{sec:back}, our methodology is detailed in \Cref{sec:meth}, experimental
results are discussed in \Cref{sec:exp}, and the paper concludes with discussions in \Cref{sec:conc}. Detailed theorems along with their proofs are provided in \Cref{sec:theos}, and supplementary plots can be found in \Cref{sec:plots}.

\section{Background}
\label{sec:back}
This section provides an overview of the foundational concepts relevant to our proposed method. \Cref{subsec:gp review} introduces GPs, GP regression, and the fundamental methodology for GP sampling. The underlying assumptions of our work are outlined in \Cref{subsec:assump}. Additionally, detailed explanations of the inducing points approximation and sparse grids, which are key components of our sampling method, are presented in \Cref{subsec:inducing} and \Cref{subsec:sparse grid}, respectively.

\subsection{GP review}
\label{subsec:gp review}
\paragraph{GPs}
A GP is a collection of random variables uniquely specified by its mean function $\mu(\cdot)$ and kernel function $K(\cdot,\cdot)$, such that any finite collection of those random variables have a joint Gaussian distribution \citep{williams2006gaussian}. We denote the corresponding GP as $\mathcal{GP}(\mu(\cdot),K(\cdot,\cdot))$.

\paragraph{GP regression}
Let $f:\bbR^d \rightarrow \bbR$ be an unknown function. Suppose we are given a set of training points $\bfX=\{\bfx_i\}_{i=1}^n$ and the observations $\bfy=(y_1,\ldots,y_n)^{\top}$ where $y_i=f(\bfx_i)+\epsilon_i$ with the i.i.d. noise $\epsilon_i \sim \calN(0,\sigma_{\epsilon}^2)$.
The GP regression imposes a GP prior over the latent function as $f(\cdot) \sim \mathcal{GP}(\mu(\cdot),K(\cdot,\cdot))$. Then the posterior distribution at $m$ test points $\bfX^*=\{\bfx_i^*\}_{i=1}^m$ is
\begin{equation}
    p(f(\bfX^*) \vert \bfy)=\calN(\bfmu_{* \vert \bfy}, \bfK_{*,* \vert \bfy}),
\end{equation}
with
\begin{subequations}\label{eq:conditional-gp}
    \begin{equation}
    \bfmu_{* \vert \bfy}= \bfmu_* + \bfK_{*,\bfX} \big[\bfK_{\bfX,\bfX} + \sigma_{\epsilon}^2 \bfI_n \big] ^{-1}\left(\bfy-\bfmu_{\bfX}\right),\label{eq:conditional-mean}
    \end{equation}
    \begin{equation}
     \bfK_{*,* \vert \bfy} = \bfK_{*,*} - \bfK_{*,\bfX} \big[\bfK_{\bfX,\bfX} + \sigma_{\epsilon}^2 \bfI_n \big]^{-1} \bfK_{\bfX,*}, \label{eq:conditional-variance}
     \end{equation}
\end{subequations}
where $\sigma_{\epsilon}^2>0$ is the noise variance, $\bfI_n$ is a $n\times n$ identity matrix, $\bfK_{*,\bfX} = K(\bfX^*,\bfX)=\big[ K(\bfx_i^*,\bfx_j) \big]_{i=1,j=1}^{i=m,j=n}=K(\bfX,\bfX^*)^{\top}=\bfK_{\bfX,*}^{\top}$,  $\bfK_{\bfX,\bfX}=\big[K(\bfx_i,\bfx_j)\big]_{i,j=1}^n$, $\bfK_{*,*}=\big[K(\bfx_i^*,\bfx_j^*)\big]_{i,j=1}^m$, $\bfmu_{\bfX}= \big(\mu(\bfx_1),\cdots,\mu(\bfx_n)\big)^{\top}$ and $\bfmu_{*}= \big(\mu(\bfx_1^*),\cdots,\mu(\bfx_m^*)\big)^{\top}$.

% \paragraph{GP training}
% The hyperparameters $\bftheta$ of GPs, which may include GP variance, kernel lengthscale, are commonly learned by maximizing the \textit{log-marginal likelihood} \citep{jones1998efficient} given by
% \begin{equation}
% \begin{aligned}
%     \calL(\bftheta) 
%     &= \log p(\bfy \vert \bfX; \bftheta)\\
%     &= -\frac{1}{2} \Big( (\bfy - \bfmu_{\bfX})^{\top} \hat{\bfK}_{\bfX,\bfX}^{-1} (\bfy - \bfmu_{\bfX}) + \log\det(\hat{\bfK}_{\bfX,\bfX})  + n \log(2\pi) \Big),
% \end{aligned}
% \end{equation}
% where $\hat{\bfK}_{\bfX,\bfX}= \bfK_{\bfX,\bfX} + \sigma_{\epsilon}^2 \bfI_n$ is the covariance matrix of all data points $\bfX=\{\bfx_{i}\}_{i=1}^n$ with diagonal observation noise. A standard way of obtaining the optimized hyperparameters is by taking the derivative of each $\theta \in \bftheta$ as follows:
% \begin{equation}
%     \frac{\partial}{\partial \theta} \calL(\bftheta) = \frac{1}{2}  (\bfy - \bfmu_{\bfX})^{\top} \hat{\bfK}_{\bfX,\bfX}^{-1} \frac{\partial \hat{\bfK}_{\bfX,\bfX}}{\partial \theta}
%     \hat{\bfK}_{\bfX,\bfX}^{-1}
%     (\bfy - \bfmu_{\bfX}) - \frac{1}{2} \text{tr}\Big( \hat{\bfK}_{\bfX,\bfX}^{-1} \frac{\partial \hat{\bfK}_{\bfX,\bfX}}{\partial \theta} \Big).
% \end{equation}

\paragraph{Sampling from GPs}
\label{subsec:gp sampling}
The goal is to sample a function $f(\cdot)\sim\mathcal{GP}(\mu(\cdot), K(\cdot,\cdot))$. To represent this in a finite form, we discretize the input space and focus on the function values at a set of grid points $\bfZ=\{\bfz_i\}_{i=1}^{n_s}$. Our goal is to generate samples $\bff_{\bfZ}=\big(f(\bfz_1),\cdots,f(\bfz_{n_s})\big)^{\top}$ from the multivariate normal distribution $\calN(\mu(\bfZ),K(\bfZ,\bfZ))=\calN(\bfmu_{\bfZ},\bfK_{\bfZ,\bfZ})$. The standard method for sampling from a multivariate normal distribution involves the following steps: (1) generate a vector of samples $\bff_{\bfI}$ whose entries are independent and identically distributed normal, i.e. $\bff_{\bfI_{n_s}}\sim\calN(\bfzero,\bfI_{n_s})$; (2) Use the Cholesky decomposition \citep{golub2013matrix} to factorize the covariance matrix $\bfK_{\bfZ,\bfZ}$ as $\bfL_{\bfZ} \bfL_{\bfZ}^
{\top} = \bfK_{\bfZ,\bfZ}$; (3) generate the output sample $\bff_{\bfZ}$ as
\begin{equation}\label{eq:prior sample chol}
    \bff_{\bfZ} \leftarrow \bfL_{\bfZ} \bff_{\bfI_{n_s}} + \bfmu_{\bfZ}.
\end{equation}

Sampling a posterior GP follows a similar procedure.
%Given a random vector $\bff_n= \big(f(\bfx_1),\cdots,f(\bfx_n)\big)^T$, $f(\cdot)$ is a Gaussian process defined in \Cref{subsec:gp review}, then it follows a multivariate normal distribution, in particular, $\bff_n \sim \calN(\mu(\bfX), K(\bfX,\bfX))=\calN(\bfmu_n, \bfK_{n,n})$.
%As for sampling from GP posteriors, assume 
Suppose we have observations $\bfy= \big(y_1,\cdots,y_n\big)^{\top}$ on $n$ distinct points $\bfX=\{\bfx_i\}_{i=1}^n$, where $y_i=f(\bfx_i)+\epsilon_i$ with noise $\epsilon_i \sim \calN(0,\sigma_{\epsilon}^2)$ and $f(\cdot) \sim \mathcal{GP}(\mu(\cdot), K(\cdot,\cdot))$. The goal is to generate posterior samples $f(\bfX^*)\mid \bfy$ at $m$ test points $\bfX^*=\{\bfx_i^*\}_{i=1}^m$. 
Since the posterior samples $f(\bfX^*) \mid \bfy \sim \calN (\bfmu_{* \vert \bfy}, \bfK_{*,* \vert \bfy})$ also follow a multivariate normal distribution, as detailed in \cref{eq:conditional-mean} and \cref{eq:conditional-variance} in \Cref{subsec:gp sampling}, we can apply Cholesky decomposition $\bfL_{* \vert \bfX} \bfL_{* \vert \bfX}^{\top} = \bfK_{*,* \vert \bfX}$ to generate GP posterior sample $\bff_{*}=f(\bfX^*)=\big(f(\bfx_1^*),\cdots,f(\bfx_m^*)\big)^{\top}$ as
\begin{equation}\label{eq:post sample chol}
    \bff_{*} \mid \bfy \leftarrow \bfL_{* \vert \bfX} \bff_{\bfI_n} + \bfmu_{* \vert \bfX},
\end{equation}
where $\bff_{\bfI_{n}} \sim \calN(\bfzero,\bfI_{n})$.

\subsection{Assumptions}
\label{subsec:assump}
A widely used kernel structure in multi-dimensional settings is ``separable'' or ``product'' kernels, which can be expressed as follows:
\begin{equation}
    \label{eq:separable-kernel}
     K(\bfx,\bfx')=\sigma^2 \prod_{j=1}^d K_0^{(j)} (x^{(j)},x'^{(j)}),
\end{equation}
where $x^{(j)}$ and $x'^{(j)}$ are the $j$-th components of $\bfx$ and $\bfx'$ respectively, $K_0^{(j)}(\cdot,\cdot)$ denotes the one-dimensional correlation function for each dimension $j$ with the variance set to one, $d$ denotes the dimension of the input points. While the product of one-dimensional kernels does not share identical smoothness properties with multi-dimensional kernels, the assumption of separability is widely adopted in the field of spatio-temporal statistics \citep{gneiting2006geostatistical,genton2007separable,constantinou2017testing}. Its widespread adoption is mainly due to its efficacy in enabling the development of valid space-time parametric models and in simplifying computational processes, particularly when handling large datasets for purposes such as inference and parameter estimation.
%because it allows for a simple construction of valid space-time parametric models and facilitates the computational procedures for large datasets in inference and parameter estimation.

In this article, we consider the inputs to be $d$-dimensional points $\bfX=\{\bfx_i\}_{i=1}^n$ where each $\bfx_i \in \bbR^d$ and $d\in \bbN^+$, and we focus on product Mat\'ern kernels that employ identical base kernels across every dimension. Mat\'ern kernel \citep{genton2001classes} is a popular choice of covariance functions in spatial statistics \citep{diggle2003introduction}, geostatistics \citep{curriero2006use,pardo2008geostatistics}, image analysis \citep{zafari2020resolving,okhrin2020new}, and other applications. The Mat\'ern kernel is defined as the following form \citep{stein1999interpolation}:
\begin{equation}\label{eq:Matern-1d}
     K(x,x')= \sigma^2 \frac{2^{1-\nu}}{\Gamma(\nu)}\left(\sqrt{2\nu}\frac{\lvert x-x'\rvert}{\omega}\right)^{\nu}K_{\nu}\left(\sqrt{2\nu}\frac{\lvert x-x'\rvert}{\omega}\right),
\end{equation}
for any $x,x'\in \mathbb{R}$, where $\sigma^2>0$ is the variance, $\nu>0$ is the smoothness parameter, $\omega>0$ is the lengthscale and $K_\nu$ is the modified Bessel function of the second kind.

\subsection{Inducing points approximation}
\label{subsec:inducing}
\citet{williams2000using} first applied Nystr\"{o}m approximation to kernel-based methods such as GPs. Following this, \citet{quinonero2005unifying} presented a novel unifying perspective, and reformulated the prior by introducing an additional set of latent variables $\bff_{\bfu}$, known as \textit{inducing variables}. These latent variables correspond to a specific set of input locations $\bfu=(u_1,\cdots,u_q)^{\top}$, referred to as \textit{inducing inputs}. In the subset of regressors (SoR) \citep{silverman1985some} algorithm, the kernels can be approximated using the following expression:
	\begin{equation}\label{eq:SoR approx}
	    K_{\text{SoR}}(\bfx_i,\bfx_j) = K(\bfx_i,\bfu) K(\bfu,\bfu)^{-1} K(\bfu,\bfx_j).
	   % K_{\text{FITC}}(\bfx_i,\bfx_j) = K_{\text{SoR}}(\bfx_i,\BFx_j) + \delta_{ij} ( K(\bfx_i,\bfx_j) - K_{\text{SoR}}(\bfx_i,\bfx_j) ),
	\end{equation}
% 	where the Kronecker delta $\delta_{ij}=1$ if $i=j$, $\delta_{ij}=0$ if $i\neq j$.
    % The approximate conditional distributions in the SoR model are given by:
    % \begin{equation}\label{eq:SoR conditional}
    %     \begin{aligned}
    %         q_{\rm SoR}(\bff_{\bfX} \vert \bff_{\bfu}) &= \calN(\bfK_{\bfX,\bfu} \bfK^{-1}_{\bfu,\bfu} \bff_{\bfu}, \bfzero ) ,\\
    %         q_{\rm SoR}(\bff_* \vert \bff_{\bfu}) &= \calN(\bfK_{*,\bfu} \bfK^{-1}_{\bfu,\bfu} \bff_{\bfu}, \bfzero ).
    %     \end{aligned}
    % \end{equation}
Employing the SoR approximation detailed in \cref{eq:SoR approx}, the form of the SoR predictive distribution $q_{\rm SoR}(\bff_* \mid \bfy)$ can be derived as follows:
\begin{equation}\label{eq:SoR pred distrib}
    q_{\rm SoR}(\bff_* \mid \bfy) =
     \calN(\bfmu_* + \sigma_{\epsilon}^{-2} \bfK_{*,\bfu} \bfSigma_{\bfu}^{-1} \bfK_{\bfu,\bfX} (\bfy-\bfmu_{\bfX}),\;
     \bfK_{*,\bfu} \bfSigma_{\bfu}^{-1}\bfK_{\bfu,*}),
\end{equation}
where we define $\bfSigma_{\bfu} = \big[ \bfK_{\bfu,\bfu} + \sigma_{\epsilon}^{-2} 
 \bfK_{\bfu,\bfX} \bfK_{\bfX,\bfu} \big]$, $\bfK_{*,\bfu} = K(\bfX^*,\bfu)=\big[K(\bfu,\bfX^*)\big]^{\top}=\bfK_{\bfu,*}^{\top}$,  $\bfK_{\bfu,\bfu}=\big[K(\bfu_i,\bfu_s)\big]_{i,s=1}^q$, $\bfmu_{\bfX}= \big(\mu(\bfx_1),\cdots,\mu(\bfx_n)\big)^{\top}$, and $\bfmu_{*}= \big(\mu(\bfx_1^*),\cdots,\mu(\bfx_m^*)\big)^{\top}$. Consequently, the computation of GP regression using inducing points can be achieved with a time complexity of $\calO(nq^2)$ for $q$ inducing points.

To quantify the approximation power in the inducing points approach, we shall use the \textit{Wasserstein metric}. Consider a separable complete metric space denoted as $(M,d)$. The \textit{$p$-Wasserstein distance} between two probability measures $\mu_1$ and $\mu_2$, on the space $M$, is defined as follows:
	\[W_p(\mu,\nu)=\left(\inf_{\gamma\in\Gamma(\mu_1,\mu_2)}\mathbb{E}_{(x,y)\sim \gamma}d^p(x,y)\right)^{1/p},\]
	for $p\in[1,+\infty)$, where $\Gamma(\mu_1,\mu_2)$ is the set of joint measures on $M\times M$ with marginals $\mu_1$ and $\mu_2$. \Cref{theo:induce approx order} in \Cref{sec:theos} shows that the inducing points approximation of a GP converges rapidly to the target GP, with a convergence rate of $\calO(n^{-\nu})$ for the Mat\'ern kernels with smoothness $\nu$, provided that the inducing points are uniformly distributed. It is worth noting that, inequality presented in \cref{eq:rate} is also applicable to multivariate GPs.
 
 Therefore, the convergence rate for multivariate inducing points approximation can be established, assuming the corresponding convergence rate of GP regression is known. A general framework to investigate the latter problem is given by \citet{wang2020prediction}.
Inspired by \Cref{theo:induce approx order}, we assume that the vector $\bfy=(y_1,\cdots,y_n)^{\top}$ is observed on $n$ distinct $d$-dimensional points $\bfX=\{\bfx_i\}_{i=1}^n$. The inducing inputs $\bfU$, as defined in \Cref{subsec:inducing}, are given by a full grid, denoted as $\bfU=\times_{j=1}^d\bfu^{(j)}$, where each $\bfu^{(j)}$ is a set of $q_j$ one-dimensional data points, $A \times B$ denotes the \textit{Cartesian product} of sets $A$ and $B$, accordingly the grid $\bfU$ has $q=\prod_{j=1}^d q_j$ points. Given the separable structure as defined in \cref{eq:separable-kernel} and utilizing the full grid $\bfU$, we can transform the multi-dimensional problem to a one-dimensional problem since the covariance matrix $\bfK_{\bfU,\bfU}\colon=K(\bfU,\bfU)$ can be represented by \textit{Kronecker products} \citep{henderson1983history,saatcci2012scalable,wilson2015kernel} of matrices over each input dimension $j$:
\begin{equation}\label{eq:kernel full grid}
    \bfK_{\bfU,\bfU}=K(\bfU,\bfU) = \bigotimes_{j=1}^d K_j(\bfu^{(j)},\bfu^{(j)}).
    % = \bigotimes_{j=1}^{d} \big( \bfphi_{\bfu^{(j)}}(\bfu^{(j)}) \bfA_{\bfu^{(j)}}^{-1} \big).
\end{equation}
However, the size of the inducing points $q$ grows exponentially as $d$ gets large, therefore we consider using \textit{sparse grids} \citep{garcke2013sparse,rieger2017sampling} to solve the curse of dimensionality.

\subsection{Sparse grids}
\label{subsec:sparse grid}
A sparse grid is built on a nested sequence of one-dimensional points $\varnothing = \calU_{j,0} \subseteq \ldots \subseteq \calU_{j,t} \subseteq \calU_{j,t+1}\subseteq \ldots$ for each dimension $j$, $j=1,\ldots,d$. A sparse grid with level $\eta$ and dimension $d$ is a design as follows:
\begin{equation}\label{eq:sgd}
\begin{aligned}
    \calU(\eta,d) &= \bigcup_{\vv{t}\in\bbG(\eta)} \calU_{1,t_1} \times \calU_{2,t_2} \times \cdots \times \calU_{d,t_d},\\
    % &:= \bigcup_{\vv{t}\in\bbG(\eta)} \calU_{\vv{t}}, \\
\end{aligned}
\end{equation}
where $\vv{t}=(t_1,\ldots,t_d)$, $\bbG(\eta) = \{ \vv{t}\in\bbN^{d} \vert \sum_{j=1}^d t_j = \eta \}$. \Cref{fig:sgdesign d=2} and \Cref{fig:sgdesign d=3} show sparse grids corresponding to the \textit{hyperbolic cross points} (bisection) for two dimensional points at levels $\eta=3,4,5,6$ and for three dimensional points at $\eta=4,5,6,7$ respectively. %Clenshaw–Curtis points (Chebyshev nodes) \citep{clenshaw1960method}. 
Hyperbolic cross points sets, defined over the interval $[0,1]$, are detailed in \cref{eq:clenshaw curtis points}:
\begin{equation}\label{eq:clenshaw curtis points}
    \calU_{j,t_j} = \left\{ \frac{i}{2^{t_j}}: i=1,\ldots, 2^{t_j}-1 \right\}. %\left\{ - \cos\Big(\pi\frac{i}{2^{t_j}}\Big): 1 \leq i \leq 2^{t_j}-1 \right\}.
\end{equation}
An operator on the sparse grid can be efficiently computed using Smolyak’s algorithm \citep{smolyak1963quadrature,ullrich2008smolyak}, which states that the predictor operator $\calP(\eta,d)$ on function $f$ with respect to the sparse grid $\calU(\eta,d)$ can be evaluated by the following formula:
\begin{equation}\label{eq:smolyak alg}
    \calP(\eta,d)(f) = \sum_{\vv{t}\in\bbP(\eta,d)} (-1)^{\eta - |\vv{t}|}  
    \binom{d-1}{\eta - |\vv{t}|} \Big( \bigotimes_{j=1}^d \calU_{j,t_j} \Big) (f),
\end{equation}
where $\vert \vv{t} \vert =\sum_{j=1}^d t_j$, $\bbP(\eta,d) = \{ \vv{t} \in \bbN^d \vert \max(d,\eta-d+1) \leq \vert \vv{t} \vert \leq \eta \}$, $\Big( \bigotimes_{j=1}^d \calU_{j,t_j} \Big)$ here represents the operator over the Kronecker product points $\bigotimes_{j=1}^d \calU_{j,t_j}$. \citet{plumlee2014fast} and \citet{yadav2022kernel} also provided the linear time computing method for GP regression based on Smolyak's algorithm. \Cref{theo:induce approx sg order} in \Cref{sec:theos} demonstrates that for a $d$-dimensional GP, the inducing points approximation converges to the target GP at the rate of $\calO( n^{-\nu} (\log n)^{(\nu+2)(d-1)+d+1} )$ for the product Mat\'ern kernels with smoothness $\nu$, given that the inducing points are sparse grids $\calU(\eta,d)$ of level $\eta$ and dimension $d$, as defined in \cref{eq:sgd}.

\begin{figure}
\centering
\includegraphics[width=\linewidth]{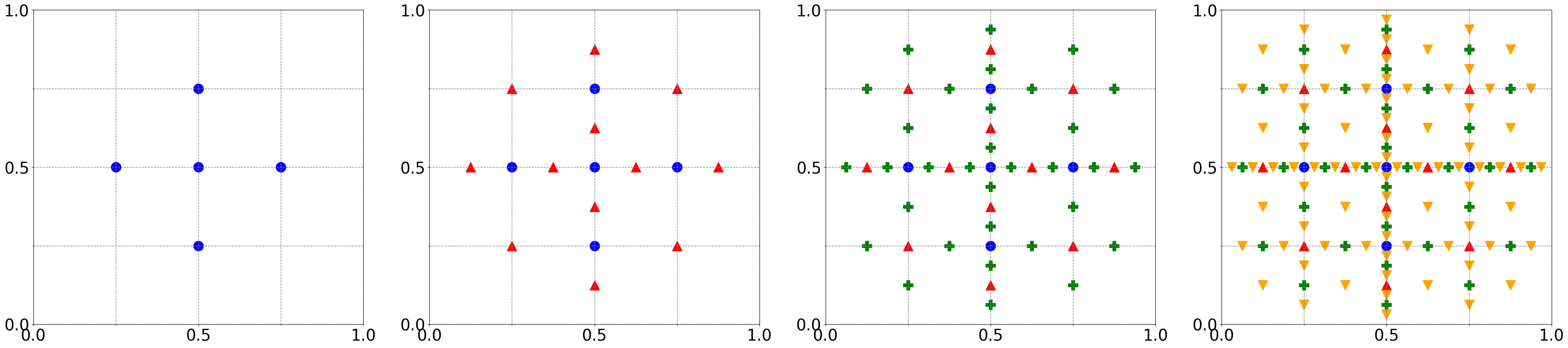}
\caption{Sparse grids $\calU(\eta,d)$ over $[0,1]^d$ of level $\eta=3,4,5,6$ and dimension $d=2$.}
\label{fig:sgdesign d=2}
\end{figure}

%\section{Methodology}
\section{Sampling with sparse grids}
\label{sec:meth}
In this section, we elaborate the algorithm for sampling from multi-dimensional GPs for product kernels with sparse grids.

\subsection{Prior}
\paragraph{Sampling algorithm}
Our goal is to sample from GP priors $f(\cdot) \sim \mathcal{GP}(\mu(\cdot), K(\cdot,\cdot))$ over $d$-dimensional grid points $\bfZ=\{\bfz_i\}_{i=1}^{n_s}$, $\bfz_i \in \bbR^d$. According to the conditional distributions of the SoR model in equations \cref{eq:SoR approx}, the distribution of $\bff_{\bfZ}=f(\bfZ)=(f(\bfz_1),\cdots,f(\bfz_{n_s}))^{\top}$ can be approximated by $ \bfK_{\bfZ,\calU} \bfK_{\calU,\calU}^{-1} \bff_{\calU}$ with $\bff_{\calU} \sim \calN(\bfzero, \bfK_{\calU,\calU})$, where $\bff_{\calU}=f(\calU)$, and $\calU=\calU(\eta,d)$ is a sparse grid design with level $\eta$ and dimension $d$ defined in \cref{eq:sgd}. Then it reduces to generating samplings from $ \bfK_{\calU,\calU}^{-1} \bff_{\calU} \sim \calN(\bfzero, \bfK_{\calU,\calU}^{-1})$. Apply Smolyak's algorithm to samplings from GPs (see \Cref{theo:sampling with sg inducing} in \Cref{sec:theos}), we can infer that
\begin{equation}\label{eq: prior sparse grid chol}
    \bfK_{\calU,\calU}^{-1} \bff_{\calU}
    \distribeq \sum_{\vv{t}\in\bbP(\eta,d)} (-1)^{\eta - |\vv{t}|} \binom{d-1}{\eta - |\vv{t}|} \bigotimes_{j=1}^d \text{chol}\Big( K(\calU_{j,t_j}, \calU_{j,t_j})^{-1} \Big) \bff_{\bfI_{n_{\vv{t}}}} \delta_{x}(\calU_{\vv{t}}),
\end{equation}
\begin{equation}\label{eq:dirac}
    \delta_{x}(\calU_{\vv{t}}) = \begin{cases}
        0 \quad & \text{if $x \in \calU \setminus \calU_{\vv{t}}$}  \\
        1 \quad & \text{if $x \in \calU_{\vv{t}}$}
    \end{cases}, %\\
\end{equation}
where $\text{chol}\Big( K(\calU_{j,t_j}, \calU_{j,t_j})^{-1} \Big)$ is the Cholesky decomposition of $K(\calU_{j,t_j}, \calU_{j,t_j})^{-1}$, $\calU=\calU(\eta,d)$ is a sparse grid design with level $\eta$ and dimension $d$ defined in \cref{eq:sgd}, $\bff_{\bfI_{n_{\vv{t}}}} \sim \calN(\bfzero,\bfI_{n_{\vv{t}}})$ is standard multivariate normal distributed, $\bfI_{n_{\vv{t}}}$ is a $n_{\vv{t}} \times n_{\vv{t}}$ identity matrix, $n_{\vv{t}}$ $=$ $\prod_{j=1}^d t_j$ is the size of the full grid design $\calU_{\vv{t}}=\calU_{1,t_1} \times \calU_{2,t_2} \times \cdots \times \calU_{d,t_d}$ associated with $\vv{t}=(t_1,\ldots,t_d)$, $\delta_{x}(\calU_{\vv{t}})$ is a Dirac measure on the set $\calU$.
%and $K(\calU_{j,t_j}, \calU_{j,t_j})^{-1} = \Big( \bfA_{\calU_{j,t_j}} \bfQ_{\calU_{j,t_j}}^{-\top} \Big) \Big( \bfA_{\calU_{j,t_j}} \bfQ_{\calU_{j,t_j}}^{-\top} \Big)^{\top}$. 
Therefore, the GP prior $\bff_{\bfZ}$ over grid points $\bfZ$ can be approximately sampled via
{\color{blue}{
\begin{align}\label{eq:multi-dim prior}
    \bff_{\bfZ} \leftarrow \bfK_{\bfZ,\calU}
    \left(
    \sum_{\vv{t}\in\bbP(\eta,d)} (-1)^{\eta - |\vv{t}|} \binom{d-1}{\eta - |\vv{t}|} \bigotimes_{j=1}^d \text{chol}\Big( K(\calU_{j,t_j}, \calU_{j,t_j})^{-1} \Big) \bff_{\bfI_{n_{\vv{t}}}}
    \delta_{x}(\calU_{\vv{t}})
    \right)\nonumber\\
    + \bfmu_{\bfZ},
\end{align}
}}where $\bff_{\bfI_{n_{\vv{t}}}} \sim \calN(\bfzero,\bfI_{n_{\vv{t}}})$ is standard multivariate normal distributed, $\delta_x(\calU_{\vv{t}})$ is a Dirac measure on set $\calU$ as defined in \cref{eq:dirac}, $\text{chol}\Big( K(\calU_{j,t_j}, \calU_{j,t_j})^{-1} \Big)$ is the Cholesky decomposition of $K(\calU_{j,t_j}, \calU_{j,t_j})^{-1}$.

\paragraph{Complexity}
Since the Kronecker matrix-vector product can be written as the form of vectorization operator (see Proposition 3.7 in \citep{kolda2006multilinear}), the Kronecker matrix-vector product in \cref{eq:multi-dim prior} only costs $\calO(\sum_{j=1}^d t_j \prod_{j=1}^d t_j)$ $=\calO(\eta n_{\vv{t}})$ time, hence the time complexity of the prior sampling via \cref{eq:multi-dim prior} is $\calO\Big( (\eta+n_s) \sum_{\vv{t}\in\bbP(\eta,d)} n_{\vv{t}} \Big)$ $= \calO \Big( (\eta+n_s) n_{sg} \Big)$, %(\eta / d)^d
where $\nu$ is the smoothness parameter of the Mat\'ern kernel, $d$ is the dimension of the sampling points $\bfZ=\{\bfz_i\}_{i=1}^{n_s}$, $n_s$ is the size of the sampling points, $n_{sg}$ is the size of inducing points on a sparse grid $\calU(\eta, d)$, $n_{\vv{t}}$ $=$ $\prod_{j=1}^d t_j$ is the size of the full grid design $\calU_{\vv{t}}=\times_{j=1}^d \calU_{j,t_j}$, $\vv{t}=(t_1,\ldots,t_d)$, $\bbP(\eta,d) = \{ \vv{t} \in \bbN^d \vert \max(d,\eta-d+1) \leq \sum_{j=1}^d t_j \leq \eta \}$.

Note that we can obtain the upper bound of the number of points in the Smolyak type sparse grids. According to Lemma 3 in \citep{rieger2017sampling}, given that $\vert\calU_{j,t_j}\vert \leq 2^{t_j}-1$, when sparse grids are constructed on hyperbolic cross points as defined in \cref{eq:clenshaw curtis points}, the cardinality of a sparse grid $\calU(\eta,d)$ defined in \cref{eq:sgd}, is thereby bounded by
\begin{align}
    n_{sg} = \vert \calU(\eta,d) \vert & \leq 1^d \big(\frac{2-1}{2}\big)^d 2^{\eta} \binom{\eta-1}{d-1} \min\{ \frac{2}{2-1}, \frac{\eta}{d} \} \nonumber\\
    & \leq 2^{\eta-d} \frac{\eta^{d-1}}{(d-1)!} \min\{ 2, \frac{\eta}{d} \}\nonumber\\
    & \leq 2^{\eta-d} \eta^{d} / (d!).
\end{align}

\subsection{Posterior}
\paragraph{Matheron's rule}
Sampling from the posterior GPs can be effectively conducted using \textit{Matheron's rule}. The Matheron's rule was popularized in the geostatistics field by \citet{journel1976mining}, and recently rediscovered by \citet{wilson2020efficiently}, who leveraged it to develop a novel approach for GP posterior samplings. Matheron's rule can be described as follows: Consider $\bfa$ and $\bfb$ as jointly Gaussian random vectors. Then the random vector $\bfa$, conditional on $\bfb=\bfbeta$, is equal in distribution to
\begin{equation}\label{eq:matheron rule}
    (\bfa\mid\bfb=\bfbeta) \distribeq \bfa + \mathrm{Cov}(\bfa,\bfb) \mathrm{Cov}(\bfb,\bfb) ^ {-1} (\bfbeta - \bfb),
\end{equation}
where $\mathrm{Cov}(\bfa,\bfb)$ denotes the covariance of $(\bfa,\bfb)$.

According to \cref{eq:matheron rule}, exact GP posterior samplers can be obtained using two jointly Gaussian random variables. Following this technique, we derive
\begin{equation}\label{thm:matheron for gp}
    \bff_*\mid \bfy \distribeq \bff_* + \bfK_{*,\bfX} \big[\bfK_{\bfX,\bfX} + \sigma_{\epsilon}^2 \bfI_n \big]^{-1} (\bfy - \bff_{\bfX} - \epsilon),
\end{equation}
where $\bff_{\bfX}=f(\bfX)$ and $\bff_{*}=f(\bfX^*)$ are jointly Gaussian random variables from the prior distribution, noise variates $\bfepsilon \sim \calN (\bfzero, \sigma_{\epsilon}^2 \bfI_n)$. Clearly, the joint distribution of $(\bff_{\bfX} ,\bff_*)$ follows the multivariate normal distribution as follows:
\begin{equation}\label{eq:prior joint}
    (\bff_{\bfX}, \bff_*) \sim \calN \Bigg(
    \begin{bmatrix}
        \bfmu_{\bfX}\\
        \bfmu_*\\
    \end{bmatrix},
    \begin{bmatrix}
        \bfK_{\bfX,\bfX} & \bfK_{\bfX,*}\\
        \bfK_{*,\bfX} & \bfK_{*,*} 
    \end{bmatrix} \Bigg).
\end{equation}
\paragraph{Sampling algorithm}
We can apply SoR approximation \cref{eq:SoR approx} to Matheron’s
rule \cref{eq:matheron rule} and obtain
{\color{blue}{
\begin{equation}\label{thm:matheron for inducing gp}
    \bff_*\mid \bfy 
    \distribapprox \bff_* + \sigma_{\epsilon}^{-2} \bfK_{*,\calU} \bfSigma_\calU^{-1} \bfK_{\calU,\bfX} (\bfy - \bff_{\bfX} - \bfepsilon),
\end{equation}
}}where $\bfSigma_\calU = \big[ \bfK_{\calU,\calU} + \sigma_{\epsilon}^{-2} \bfK_{\calU,\bfX} \bfK_{\bfX,\calU} \big]$, $\bff_{\bfX}=f(\bfX)$ and $\bff_*=f(\bfX^*)$ are jointly Gaussian random variables from the prior distribution, noise variates $\bfepsilon \sim \calN (\bfzero, \sigma_{\epsilon}^2 \bfI_n)$.
The prior resulting from the SoR approximation can be easily derived from \cref{eq:SoR approx}, yielding the following expression:
\begin{equation}\label{eq:SoR prior}
    q_{\rm SoR}(\bff_{\bfX},\bff_*) = \calN(\bfmu_{\bfX+*}, \; \bfK_{\bfX+*, \calU} \bfK_{\calU,\calU}^{-1} \bfK_{\calU, \bfX+*} ),
\end{equation}
where $\bfmu_{\bfX+*}=[\bfmu_{\bfX}, \bfmu_*]^{\top}$, $\bfK_{\bfX+*, \calU} = [K(\calU,\bfX),K(\calU,\bfX^*)]^{\top}$. The same as above in \cref{eq:multi-dim prior}, $(\bff_{\bfX},\bff_*)$ can be sampled from
{\color{blue}{
\begin{align}\label{eq:multi-dim prior for posteriors}
    (\bff_{\bfX},\bff_*) 
    \leftarrow
    \bfK_{\bfX+*,\calU}
    \left(
    \sum_{\vv{t}\in\bbP(\eta,d)} (-1)^{\eta - |\vv{t}|} \binom{d-1}{\eta - |\vv{t}|} \bigotimes_{j=1}^d \text{chol}\Big( K(\calU_{j,t_j}, \calU_{j,t_j})^{-1} \Big) \bff_{\bfI_{n_{\vv{t}}}}
    \delta_{x}(\calU_{\vv{t}})
    \right)\nonumber\\
    +\bfmu_{\bfX+*},
\end{align}
}}
where $\bff_{\bfI_{n_{\vv{t}}}}$, $\delta_x(\calU_{\vv{t}})$, and $\text{chol}\Big( K(\calU_{j,t_j}, \calU_{j,t_j})^{-1} \Big)$ are the same as that defined in \cref{eq:multi-dim prior}. 

\paragraph{Conjugate gradient}

\begin{algorithm}[hbt!]
\caption{The Preconditioned Conjugate Gradient (PCG) Algorithm.}
\label{alg:pcg}
\setstretch{0.99} % set the line spacing to 0.99
\begin{algorithmic}[1]
    \STATE {\bfseries Input:} matrix $\bfSigma$, preconditioning matrix $\bfP$, vector $\bfv$, convergence threshold $\epsilon$, initial vector $\bfx_0$, maximum number of iterations $T$
    \STATE {\bfseries Output:} $\bfx_* \approx \bfSigma^{-1} \bfv$
    \STATE $\bfr_0:=\bfv-\Sigma\bfx_0$; $\bfz_0:=\bfP^{-1}\bfr_0$ \STATE $\bfp_0:=\bfz_0$
    \FOR{$i=0$ {\bfseries to} $T$}
        \STATE $\alpha_i:=\frac{\bfr_i^T \bfz_i}{\bfp_i^T \bfSigma \bfp_i}$
        \STATE $\bfx_{i+1}:=\bfx_i + \alpha_i \bfp_i$
        \STATE $\bfr_{i+1}:= \bfr_i - \alpha_i \bfSigma \bfp_i$
        \IF{$\Vert\bfr_{i+1}\Vert <\epsilon$}
            \STATE return  $\bfx_*:=\bfx_{i+1}$
        \ENDIF  
        \STATE $\bfz_{i+1}:=\bfP^{-1}\bfr_{i+1}$
        \STATE $\beta_i:=\frac{\bfr_{i+1}^T \bfz_{i+1}}{\bfr_i^T \bfz_i}$
        \STATE $\bfp_{i+1}:=\bfz_{i+1}+\beta_i \bfp_i$
    \ENDFOR
\end{algorithmic}
\end{algorithm}

The only remaining computationally demanding step in \cref{thm:matheron for inducing gp} is the computation of $\bfSigma_{\calU}^{-1}\bfv$ with $\bfv=\bfK_{\calU,\bfX} (\bfy - \bff_{\bfX} - \bfepsilon)$, requiring a time complexity of $\calO(n_{sg}^3)$ where $n_{sg}= \vert \calU(\eta,d) \vert$ is the size of the sparse grid $\calU(\eta,d)$. To reduce this computational intensity, we consider using the \textit{conjugate gradient} (CG) method \citep{hestenes1952methods,golub2013matrix}, an iterative algorithm for solving linear systems efficiently via matrix-vector multiplications. Preconditioning \citep{trefethen1997numerical,demmel1997applied,saad2003iterative,van2003iterative} is a well-known tool for accelerating the convergence of the CG method, which introduces a symmetric positive-definite matrix $\bfP_{\calU} \approx \bfSigma_{\calU}$ such that $\bfP_{\calU}^{-1} \bfSigma_{\calU}$ has a smaller condition number than $\bfSigma_{\calU}$. The entire scheme of the \textit{preconditioned conjugate gradient} (PCG)  takes the form in \Cref{alg:pcg}.

\paragraph{Preconditioner choice}
The choice of the precondition $\bfP_{\calU}$ has a trade-off between the smallest time complexity of applying $\bfP_{\calU}^{-1}$ operation and the optimal condition number of $\bfP_{\calU}^{-1} \bfSigma_{\calU}$. \citet{cutajar2016preconditioning} first investigated several groups of preconditioners for radial basis function kernel, then \citet{gardner2018gpytorch} proposed a specialized precondition for general kernels based on pivoted Cholesky decomposition. \citet{wenger2022preconditioning} analyzed and summarized convergence rates for arbitrary multi-dimensional kernels and multiple preconditioners. %However, common classes of preconditioners may not apply to $\bfSigma_{\calU}$ effectively 
Due to the hierarchical structure of the sparse grid $\calU$, instead of using common classes of preconditioners, we consider \textit{two-level additive Schwarz} (TAS) preconidtioner \citep{toselli2004domain} for the system matrix $\bfSigma_{\calU}$, which is formed as an \textit{additive Schwarz} (AS) preconditioner \citep{dryja1994domain,cai1999restricted} coupled with an additive coarse space correction. Since the sparse grid $\calU(\eta,d)$ defined in \cref{eq:sgd} is covered by a set of overlapping domains $\{ \calU_{\vv{t}} \}_{\vv{t} \in \bbG(\eta) }$ where $\bbG(\eta) = \{ \vv{t}\in\bbN^d \vert \sum_{j=1}^d t_j = \eta \}$, the TAS precondition for the system matrix $\bfSigma_{\calU}$ is then defined as
\begin{equation}\label{eq:TAS}
    \bfP_{\text{TAS}}^{-1} :=  \sum_{ \vv{t} \in \bbG(\eta)} \bfS_{\vv{t}}^{\top} \bfP_{\vv{t}}^{-1} \bfS_{\vv{t}} + \bfS_c^{\top} \bfP_c^{-1} \bfS_c.
\end{equation}
On the right hand side of \cref{eq:TAS}, the first term  is the one-level AS preconditioner, and the second term is the additive coarse space correction. $\bfS_{\vv{t}} \in \bbR^{ n_{\vv{t}} \times n_{sg} }$ is the selection matrix that projects $\calU_{\vv{t}}$ to $\calU$, $\bfP_{\vv{t}} = \bfS_{\vv{t}} \bfSigma_{\calU} \bfS_{\vv{t}}^{\top} \in \bbR^{ n_{\vv{t}} \times n_{\vv{t}} }$ is a symmetric positive definite sub-matrix defined on $\calU_{\vv{t}}$. $\bfS_c \in \bbR^{ n_c \times n_{sg} }$ is the selection matrix whose columns spanning the coarse space, and $\bfP_c = \bfS_c \bfSigma_{\calU} \bfS_c^{\top} \in \bbR^{ n_c \times n_c }$ is the coarse space matrix. $n_c$ is the dimension of the coarse space, $n_{sg}:=\vert \calU(\eta,d) \vert$ is the size of the sparse grid $\calU(\eta,d)$, $n_{\vv{t}}:=\vert \calU_{\vv{t}} \vert =\prod_{j=1}^d t_j$ is the size of the full grid design $\calU_{\vv{t}}=\times_{j=1}^d \calU_{j,t_j}$, $\vv{t}=(t_1,\ldots,t_d)$. \Cref{fig:diagram d=2} illustrates an example of the selection matrices $\{ \bfS_{\vv{t}} \}_{\vv{t} \in \bbG(\eta)}$ for the sparse grid $\calU(\eta=3,d=2)$.

\begin{figure*}
    \centering
    \includegraphics[width=\linewidth]{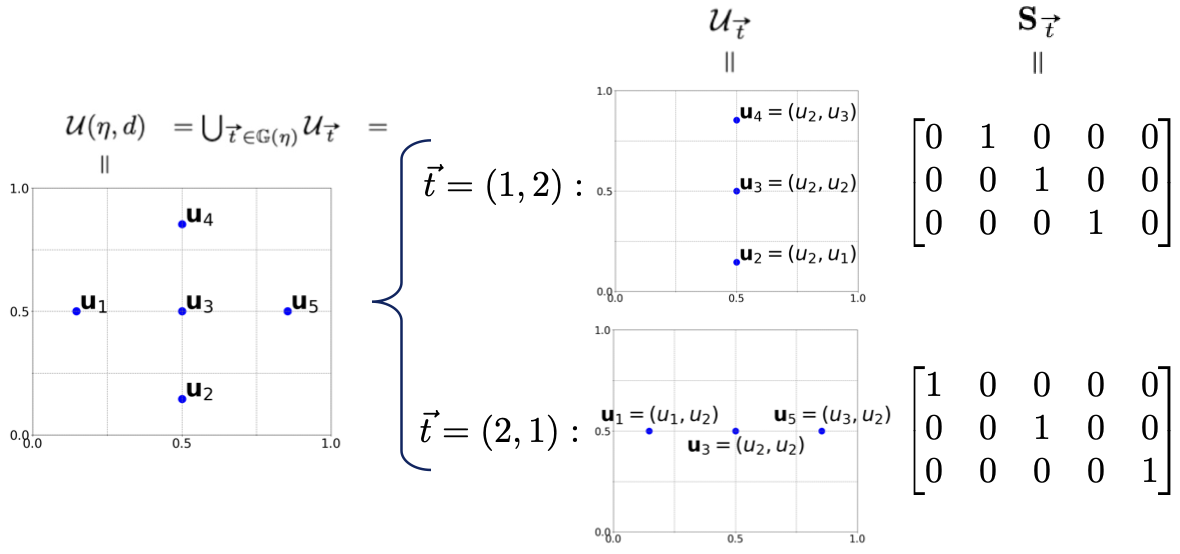}
    \caption{Selection matrices with respect to the sparse grid $\calU(\eta=3,d=2)$.}
    \label{fig:diagram d=2}
\end{figure*}

Some recent work  \citep{toselli2004domain,spillane2014abstract,dolean2015introduction,al2023efficient} provide robust convergence estimates of TAS preconditioner, which states that the condition number of $\bfP_{\text{TAS}}^{-1} \bfSigma_{\calU}$ is bounded by the number of distinct colors required to color the graph of $\bfSigma_{\calU}$ so that $span\{ \bfS_{\vv{t}} \}_{ \vv{t} \in \bbG(\eta) }$ of the same color are mutually $\bfSigma_{\calU}$-orthogonal, the maximum number of subdomains that share an unknown, and a user defined tolerance $\tau$. Specifically, given a threshold $\tau>0$, a GenEO coarse space (see Section 3.2 in \citep{al2023efficient}) can be constructed, and the following inequality holds for the preconditioned matrix $\bfP_{\text{TAS}}^{-1} \bfSigma_{\calU}$:
\begin{equation}\label{ineq: tas cond}
    \kappa( \bfP_{\text{TAS}}^{-1} \bfSigma_{\calU} ) \leq (k_c + 1) \left( 2 + (2k_c+1) \frac{k_m}{\tau} \right),
\end{equation}
where $k_c$ is the minimum number of distinct colors so that each two neighboring subdomains among $\{ \calU_{\vv{t}} \}_{\vv{t}\in\bbG(\eta)}$ have different colors, $k_m$ is the coloring constant which is the maximum number of overlapping subdomains sharing a row of $\bfSigma_{\calU}$. Since each point in $\calU(\eta,d)$ belongs to at most $\vert \bbG(\eta) \vert$ of the subdomains, $k_c=\vert \bbG(\eta) \vert$. Coloring constant $k_m$ can be viewed as the number of colors needed to color each subdomain in such a way that any two subdomains with the same color are orthogonal, therefore $k_m=\vert \bbG(\eta) \vert$.

However, the conventional approach for constructing a coarse space to bound the smallest eigenvalue of the $\bfP_{\text{TAS}}^{-1} \bfSigma_{\calU}$ relies on the \textit{algebraic local symmetric positive semidefinite} (SPSD) splitting \citep{al2019class}, which is computationally complicated. To select a coarse space that is practical and less complicated, we propose constructing it based on the sparse grid $\calU(\eta_c,d)$ with a lower level $\eta_c < \eta$. In practice, we set $\eta_c=\max\{ \lceil \eta/2 \rceil, d\}$. Although the TAS preconditioners with our proposed coarse space do not theoretically guarantee robust convergence, empirical results presented in \Cref{fig:precond} indicate that TAS tends to converge faster than one-level AS preconditioner and CG method without preconditioning. \Cref{fig:precond blowup} shows the cases where $\Sigma_{\calU}$ is ill-conditioned, resulting in the instability (blow-up) of the CG method without any preconditioning.

\begin{figure}
\centering
\begin{subfigure}[b]{0.48\textwidth}
    \includegraphics[width=\linewidth]{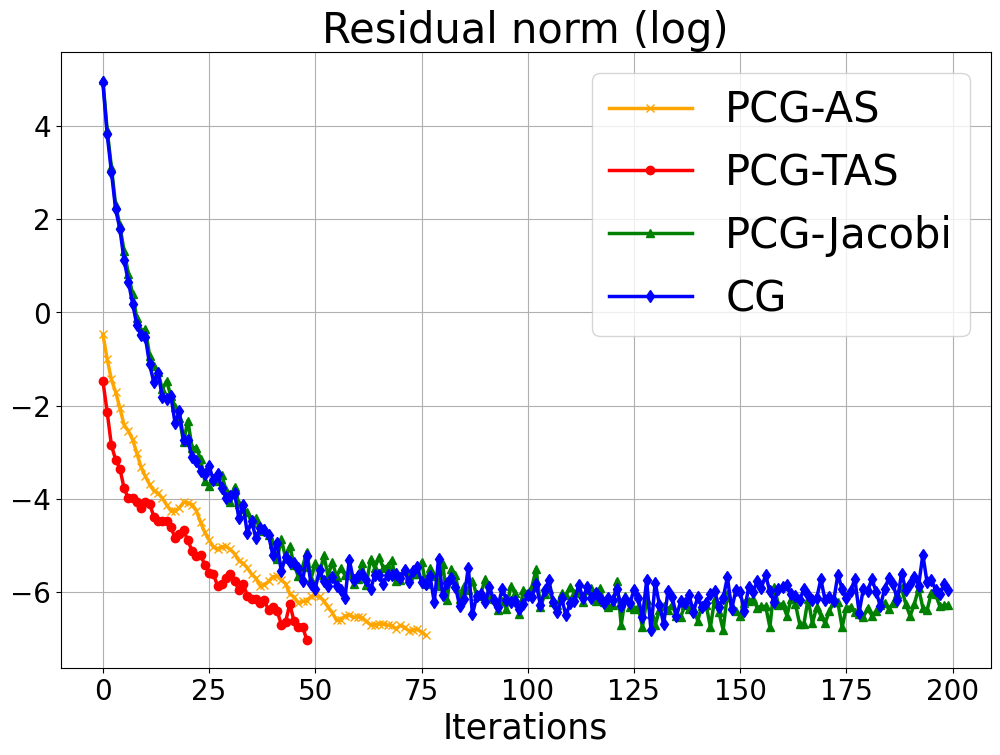}
\end{subfigure} 
\begin{subfigure}[b]{0.48\textwidth}
    \includegraphics[width=\linewidth]{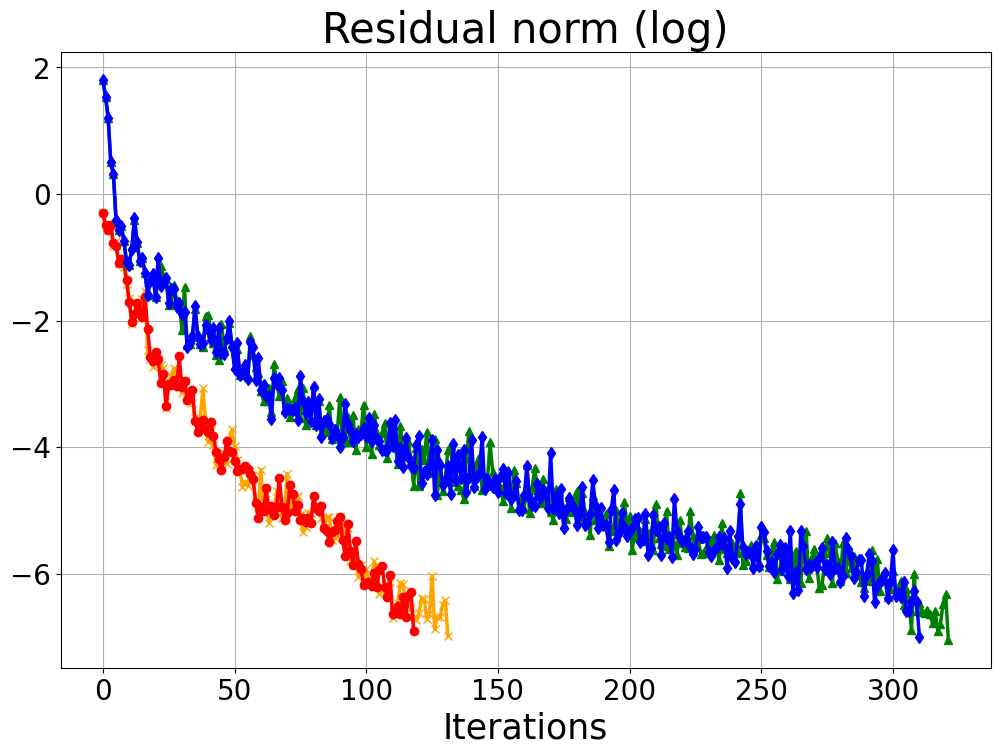}
\end{subfigure}
\caption{Residuals of PCG for $\Sigma_{\calU}^{-1} \bfv$ for tolerance $\tau=10^{-3}$ with different preconditioners where $\Sigma_{\calU}= \big[ \bfK_{\calU,\calU} + \sigma_{\epsilon}^{-2} \bfK_{\calU,\bfX} \bfK_{\bfX,\calU} \big]$, $\calU=\calU(\eta,d)$ is the sparse grid defined in \cref{eq:sgd}, vector $\bfv$ is randomly generated. $x$-axis is the iteration step, $y$-axis is the logarithm of the residuals norm. We denote PCG with TAS preconditioners by red dots, PCG with AS preconditioners by orange crosses, PCG with Jacobi preconditioners by green triangles, CG without preconditioning by blue diamonds. \textit{Left:} The sparse grid $\calU(\eta=12,d=2)$. \textit{Right:} The sparse grid $\calU(\eta=10,d=4)$.}
\label{fig:precond}
\end{figure}

\paragraph{Complexity} Note that $\bfS_{\vv{t}}$ is a sparse matrix with $n_{\vv{t}}$ nonzeros and all of its nonzeros are one, hence computing $\bfP_{\vv{t}}$ requires $\calO(n n_{\vv{t}}^2)$ operations, and the computation of the first term of $\bfP_{\text{TAS}}^{-1}$ costs $\calO( \sum_{ \vv{t} \in \bbG(\eta)} n n_{\vv{t}}^2 )$. The computational complexity of the preconditioner $\bfP_{\text{PAS}}^{-1}$ is then $\calO( \sum_{ \vv{t} \in \bbG(\eta)} n n_{\vv{t}}^2 + n_c^2) = \calO(n |\bbG(\eta)| (\eta / d)^{2d}) = \calO(n n_{sg} (\eta / d)^{2d})$ where $n_c=|\calU(\eta_c,d)|$ is the size of the sparse grid according to the former setting.

The complexity of the matrix-vector multiplication $\Sigma \bfx_0$ in \Cref{alg:pcg} is $\calO(n n_{sg})$ if $\bfSigma=\bfSigma_{\calU}$. In theory, the maximum number of iterations of the PCG with error factor $\epsilon$ is bounded by $\lceil \frac{1}{2}\sqrt{\kappa} \ln \big( \frac{2}{\epsilon}\big) \rceil$ \citep{shewchuk1994introduction}, where $\kappa$ is the condition number of the preconditioned matrix $\bfP_{\text{TAS}}^{-1} \bfSigma_{\calU}$, so the time complexity of \Cref{alg:pcg} for $\bfSigma=\bfSigma_{\calU}$ is $\calO(( \sqrt{\kappa} + (\eta/d)^{2d} )  n_{sg} n )$ where $\kappa = \kappa(\bfP_{\text{TAS}}^{-1} \bfSigma_{\calU})$. Therefore the time complexity of the entire posterior sampling scheme is $\calO \Big( \big(\eta + (\sqrt{\kappa} + (\eta/d)^{2d} + m) n \big) n_{sg} \Big)$, where $\nu$ is the smoothness parameter of the Mat\'ern kernel, $n_{sg}$ is the size of inducing points on a sparse grid $\calU(\eta,d)$, $n$ is size of observations $\bfX=\{\bfx_i\}_{i=1}^n$, $m$ is the size of test points $\bfX^*=\{\bfx_i^*\}_{i=1}^m$, $d$ is the dimension of the observations.

\section{Experiments}
\label{sec:exp}
In this section, we will demonstrate the computational efficiency and accuracy of the proposed sampling algorithm. We first generate samples from GP priors and posteriors for dimension $d=2$ and $d=4$ in \Cref{subsec:simu}. Then we apply our method to the same real-world problems as those addressed in \citep{wilson2020efficiently}, described in \Cref{subsec:appl}. For prior samplings, we use the \textit{random Fourier features} (RFF) with $2^6=64$ features \citep{rahimi2007random} and the Cholesky decomposition method as benchmarks. For posterior samplings, we compare our approach with the decoupled algorithm \citep{wilson2020efficiently} using RFF priors and the exact Matheron's update, as well as the Cholesky decomposition method. We use Mat\'ern kernels as defined in \cref{eq:Matern-1d} with the variance $\sigma^2=1$, lengthscale $\omega = \sqrt{2\nu}$ and smoothness $\nu=3/2$. Our experiments employ ``separable'' Mat\'ern kernels as specified in \cref{eq:separable-kernel}, with the same base kernels and the same lengthscale parameters $\omega=\sqrt{2\nu}$ in each dimension. For the proposed method, the inducing points are selected as a sparse grid $\calU(\eta,d)$ as defined in \cref{eq:sgd} for $(\eta=5,d=2)$, and $(\eta=6,d=4)$ respectively. We set the same nested sequence $\varnothing = \calU_{j,0} \subseteq \ldots \subseteq \calU_{j,t} \subseteq \calU_{j,t+1}\subseteq \ldots$ for each dimension $j$, and set $\calU_{j,t_j}$ as hyperbolic cross points defined in \cref{eq:clenshaw curtis points} for all $j$. We set the noise variates $\epsilon \sim \calN(\bfzero,\sigma_{\epsilon}^2 \bfI)$ with $\sigma_{\epsilon}^2=10^{-4}$ for all experiments. The seed value is set to $99$, and each experiment is replicated $1000$ times.

\subsection{Simulation}
\label{subsec:simu}
\paragraph{Prior sampling}
We generate prior samples $\bfZ=\{z_i\}_{i=1}^{n_s}$ over the domain $[0,1]^d$ using Mat\'ern 3/2 kernels, with $n_s=2^6,\ldots,2^{13}$ respectively. Left plots in \Cref{fig:prior d=2} and \Cref{fig:prior d=4} show the time required by different algorithms to sample from the different number of points $n_s$. We can observe that the proposed algorithm (InSG) is comparable in efficiency to the RFF method. To evaluate accuracy, we compute the 2-Wasserstein distances between empirical and true distributions. The 2-Wasserstein distance measures the similarity of two distributions. Let $f_1\sim\calN(\mu_1,\Sigma_1)$ and $f_2\sim\calN(\mu_2,\Sigma_2)$, the 2-Wasserstein distance between the Gaussian distributions $f_1$, $f_2$ on $L^2$ norm is given by \citep{dowson1982frechet, gelbrich1990formula,mallasto2017learning}
\begin{equation}\label{eq:2-Wasserstein}
    \begin{aligned}
     W_{2}(f_1,f_2) := & \Big( ||\mu_1 - \mu_2||^2 + {\rm tr}\big(\Sigma_1 + \Sigma_2 - 2(\Sigma_1^{\frac{1}{2}}\Sigma_2\Sigma_1^{\frac{1}{2}})^{\frac{1}{2}}\big) \Big)^{\frac{1}{2}}.
    \end{aligned}
\end{equation}
For the empirical distributions, the parameters are estimated from the replicated samples. The right plots in \Cref{fig:prior d=2} and \Cref{fig:prior d=4} demonstrate that the proposed algorithm performs with nearly the same precision as Cholesky decomposition
and RFF.
\begin{figure}[htb]
\centering
\begin{subfigure}[b]{0.48\textwidth}
    \includegraphics[width=\linewidth]{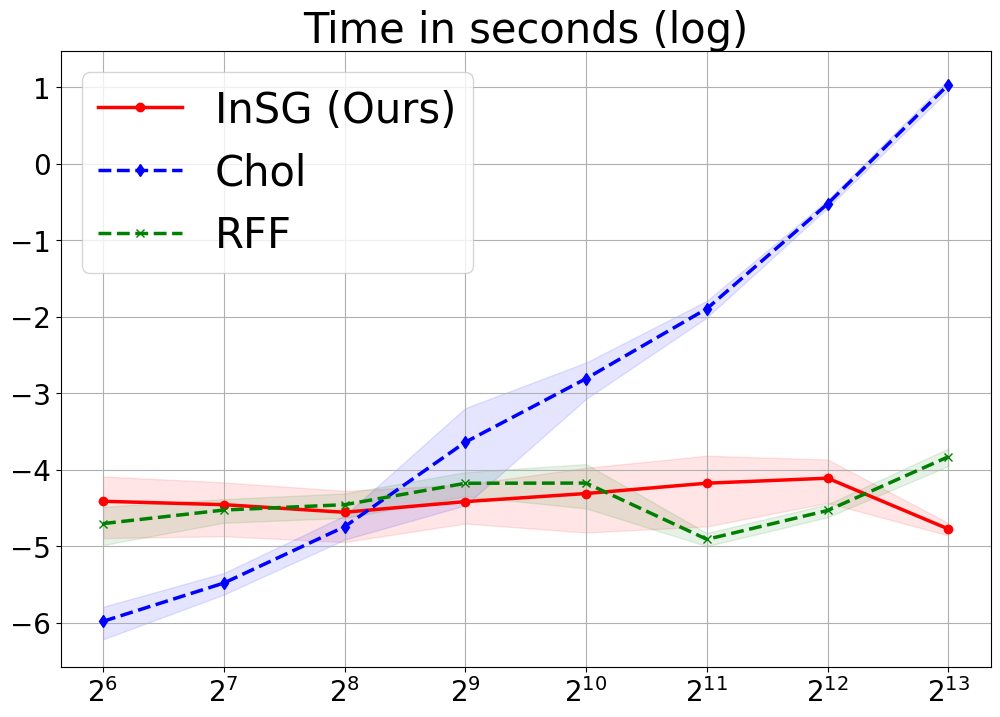}
  %\label{}%\caption{}
\end{subfigure} 
\begin{subfigure}[b]{0.48\textwidth}
    \includegraphics[width=\linewidth]{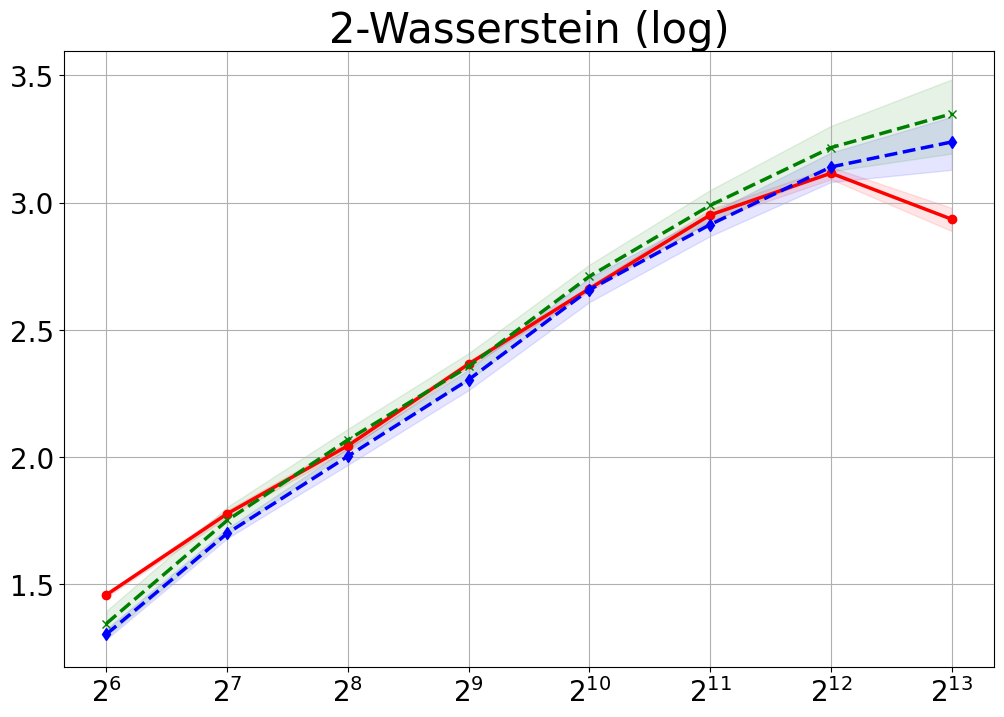}
    %\label{}%\caption{}
\end{subfigure}
\caption{Time and accuracy of different algorithms for sampling from GP priors with Mat\'ern 3/2 in dimension $d=2$. $x$-axis is the number of the sampling points $n_s$. The sparse grid level of the proposed method is set as $\eta=5$. We denote the proposed approach, the inducing points on the sparse grid (InSG), by red dots, random Fourier features (RFF) by green triangles, Cholesky decomposition (Chol) by blue diamonds. \textit{Left:} Logarithm of time taken to generate a draw from GP priors. \textit{Right:} Logarithm of 2-Wasserstein distances between priors and true distributions.}
\label{fig:prior d=2}
\end{figure}

\begin{figure}[htb]
\centering
\begin{subfigure}[b]{0.48\textwidth}
    \includegraphics[width=\linewidth]{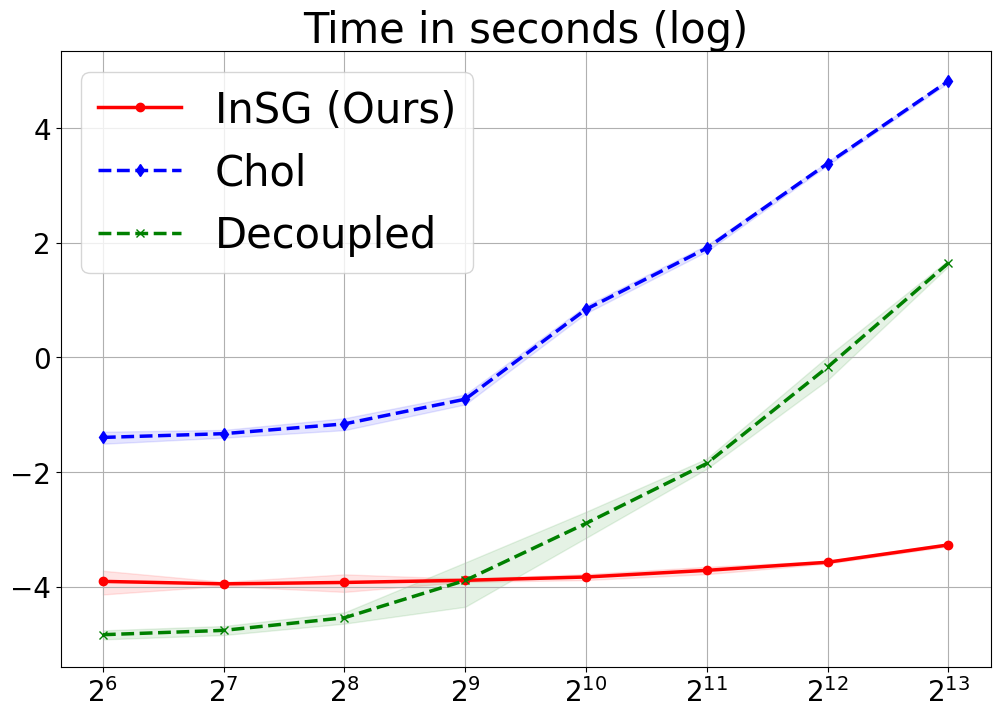}
\end{subfigure} 
\begin{subfigure}[b]{0.48\textwidth}
    \includegraphics[width=\linewidth]{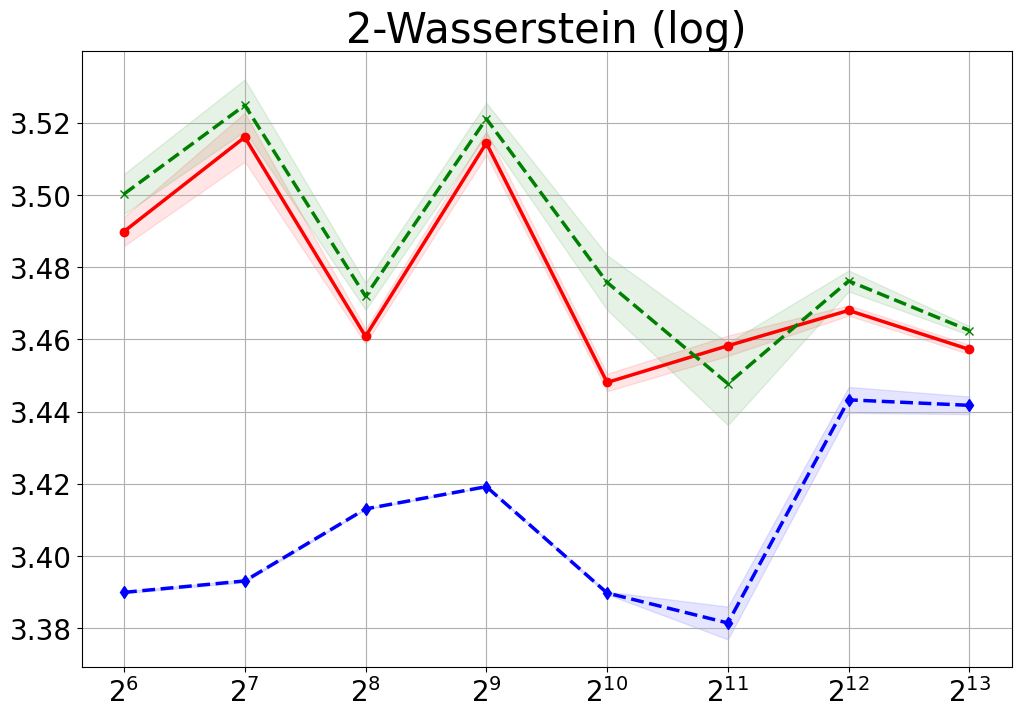}
\end{subfigure}
\caption{Time and accuracy of different algorithms for sampling from GP posteriors with Mat\'ern 3/2 and dimension $d=2$. $x$-axis is the number of observations $n$. The sparse grid level of the proposed method is set as $\eta=5$. The decoupled method is configured to use RFF priors with $2^8=64$ features and exact Matheron's update. \textit{Left:} Logarithm of time taken to generate a draw from GP posteriors over $m=1000$ points. \textit{Right:} Logarithm of 2-Wasserstein distances between posterior samplings and ground truth GP posteriors over $m=1000$ points.}
\label{fig:posterior d=2}
\end{figure}

\paragraph{Posterior sampling}
We use the Griewank function \citep{griewank1981generalized} as the test function, defined as:
\begin{equation*}
    f(\bfx)=\sum_{j=1}^d \frac{x_j^2}{4000} + \prod_{j=1}^d\cos(\frac{x_j}{\sqrt{j}}) + 1,\quad \bfx\in[-5,5]^d.
\end{equation*}
 The sparse grid $\calU(\eta,d)$ is set over $[-5,5]^d$ with configurations $(\eta=5,d=2)$ and $(\eta=6,d=4)$ for the design of our experiment. We then evaluate the average computational time and 2-Wasserstein distance over $m=1000$ random test points for each sampling method. \Cref{fig:posterior d=2} and \Cref{fig:posterior d=4} illustrate the performance of different sampling strategies. It's observed that our proposed algorithm is significantly more time-efficient compared to the Cholesky and decoupled algorithms, particularly when $n$ is larger than $2^{10}$. Moreover, our proposed method, the inducing points approximation on the sparse grid (InSG) demonstrates comparable accuracy to the decoupled method.

\subsection{Application}
\label{subsec:appl}
\paragraph{Thompson sampling} \textit{Thompson Sampling} (TS) \citep{thompson1933likelihood} is a classical strategy for decision-making by selecting actions $x\in \calX$ that minimize a black-box function $f:\calX \rightarrow \bbR$. At each iteration $t$, TS determines $\bfx_{t+1}\in \arg\min_{\bfx\in\calX}(f\vert \bfy)(\bfx)$, where $\bfy$ represents the observation set at current iteration. Upon finding the minimizer, the corresponding $y_{t+1}$ is obtained by evaluating $f$ at $\bfx_{t+1}$, and then the pair $(\bfx_{t+1},y_{t+1})$ is added to the training set. In this experiment, we consider the Ackley function \citep{ackley2012connectionist}
\begin{equation}\label{eq:ackley}
    \begin{aligned}
    f(\bfx) = -a \exp \Big( -b \sqrt{\frac{1}{d} \sum_{j=1}^d x_{j}^2} \Big) + a + \exp(1)
     - \exp \Big( \frac{1}{d} \sum_{j=1}^d\cos(c x_j) \Big),
    \end{aligned}
\end{equation}
with $a=20$, $b=0.2$, $c=2\pi$.
The goal of this experiment is to find the global minimizer of the target function $f$. We start with 3 initial samples, then at each iteration of TS, we draw a posterior sample $f \vert \bfy$ on 1024 uniformly distributed points over the interval $[-5,5]^d$ conditioned on the current observation set. Next, we pick the smallest posterior sample for this iteration, add it to the training set, and repeat the above process. This iterative approach helps us progressively approximate the global minimum. To neutralize the impact of the lengthscales and seeds, we average the regret over 8 different lengthscales $[0.2, 0.4, 0.6, 0.8, 1.0, 1.2, 1.4, 1.6]$ and 5 different seeds $[9, 99, 999, 9999, 99999]$. In \Cref{fig:ts regret}, we compare the logarithm of regret at each iteration for different sampling algorithms, the proposed approach (InSG) demonstrates comparable performance to the decoupled method.

\begin{figure*}[hbt!]
\centering
\begin{subfigure}[b]{0.325\textwidth}
  \includegraphics[width=\textwidth]{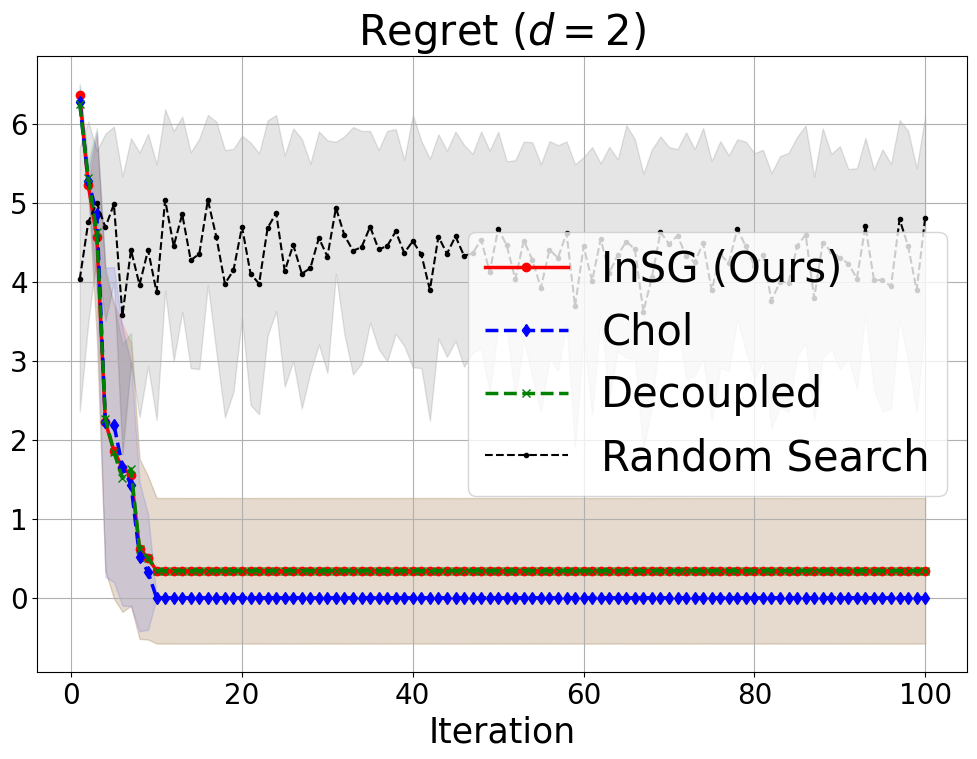}
\end{subfigure}
\begin{subfigure}[b]{0.325\textwidth}
  \includegraphics[width=\textwidth]{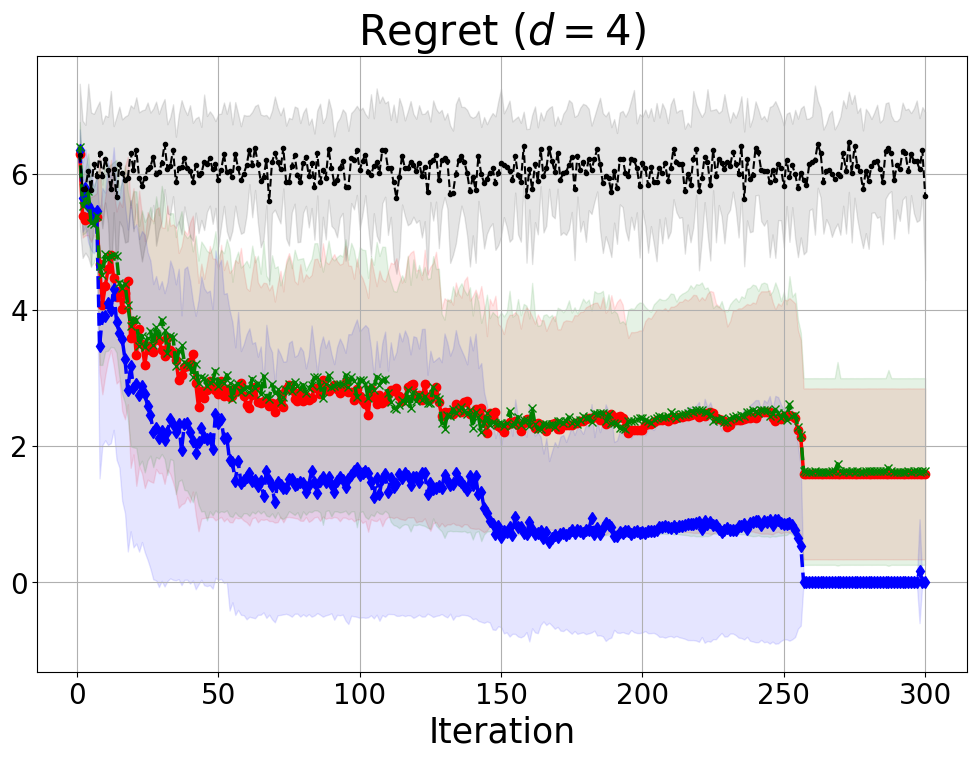}
\end{subfigure}
\begin{subfigure}[b]{0.325\textwidth}
  \includegraphics[width=\textwidth]{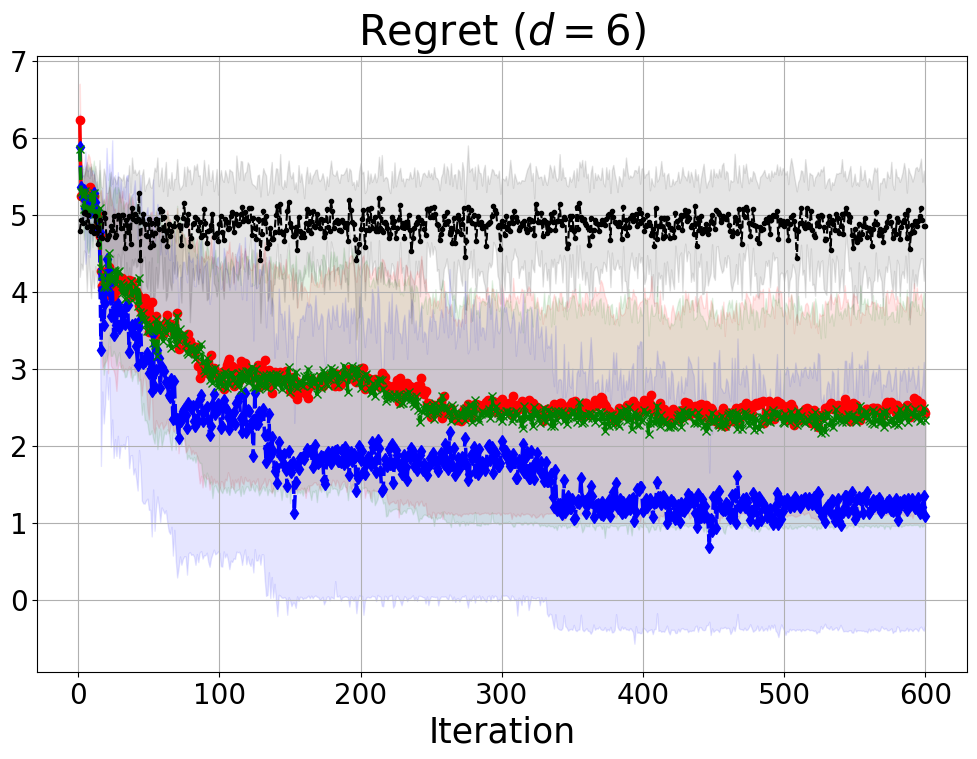}
\end{subfigure}
\caption{Regret of Thompson sampling methods of different algorithms for optimizing Ackley function with $d=2$ (\textit{left}), $d=4$ (\textit{middle}) and $d=6$ (\textit{right}) for Mat\'ern 3/2. The sparse grid design is set as $\calU(\eta=5,d=2)$, $\calU(\eta=6,d=4)$, and $\calU(\eta=8,d=6)$ respectively. The decoupled method is configured to use RFF priors with $2^8=64$ features and exact Matheron's update. $x$-axis is the iteration step, $y$-axis is the regret at each step.}
\label{fig:ts regret}
\end{figure*}

\paragraph{Simulating dynamical systems}
GP posteriors are also useful in dynamic systems where data is limited. An illustrative example is the FitzHugh–Nagumo model \citep{fitzhugh1961impulses,nagumo1962active}, which intricately models the activation and deactivation dynamics of a spiking neuron. The system's equations are given by:
\begin{equation}\label{eq:FNmodel}
    \dot{\bfx}= f(\bfx,a) = \begin{bmatrix}
        \dot{v}\\
        \dot{w}
    \end{bmatrix}=\begin{bmatrix}
        v -\frac{v^3}{3} - w + a\\
        \frac{1}{\gamma} (v-\beta w + \alpha)
    \end{bmatrix}.
\end{equation}
Using the Euler–Maruyama method \citep{maruyama1955}, the system's equations \cref{eq:FNmodel} can be discretized into a stochastic difference equation as follows:
\begin{equation}\label{eq:sde}
    \bfy_t = \bfx_{t+1} - \bfx_t = \tau f(\bfx_t,a_t) + \sqrt{\tau} \epsilon_t, \quad \epsilon_t \sim \calN(\bfzero,\sigma_{\epsilon}^2 \bfI),
 \end{equation}
where $\tau$ is the fixed step size. Write a FitzHugh–Naguomo neuron \cref{eq:FNmodel} in the form \cref{eq:sde}, we can obtain:
\begin{equation}\label{eq:FNmodel sde}
    \bfx_{t+1}-\bfx_{t} =\tau \begin{bmatrix}
        v_t - \frac{v_t^3}{3}-w_t + a_t\\
        \frac{1}{\gamma} (v_t - \beta w_t +\alpha)
    \end{bmatrix}
    + \sqrt{\tau} \epsilon_t.
\end{equation}
\begin{figure*}[htb]
    \centering
    \includegraphics[width=\linewidth]{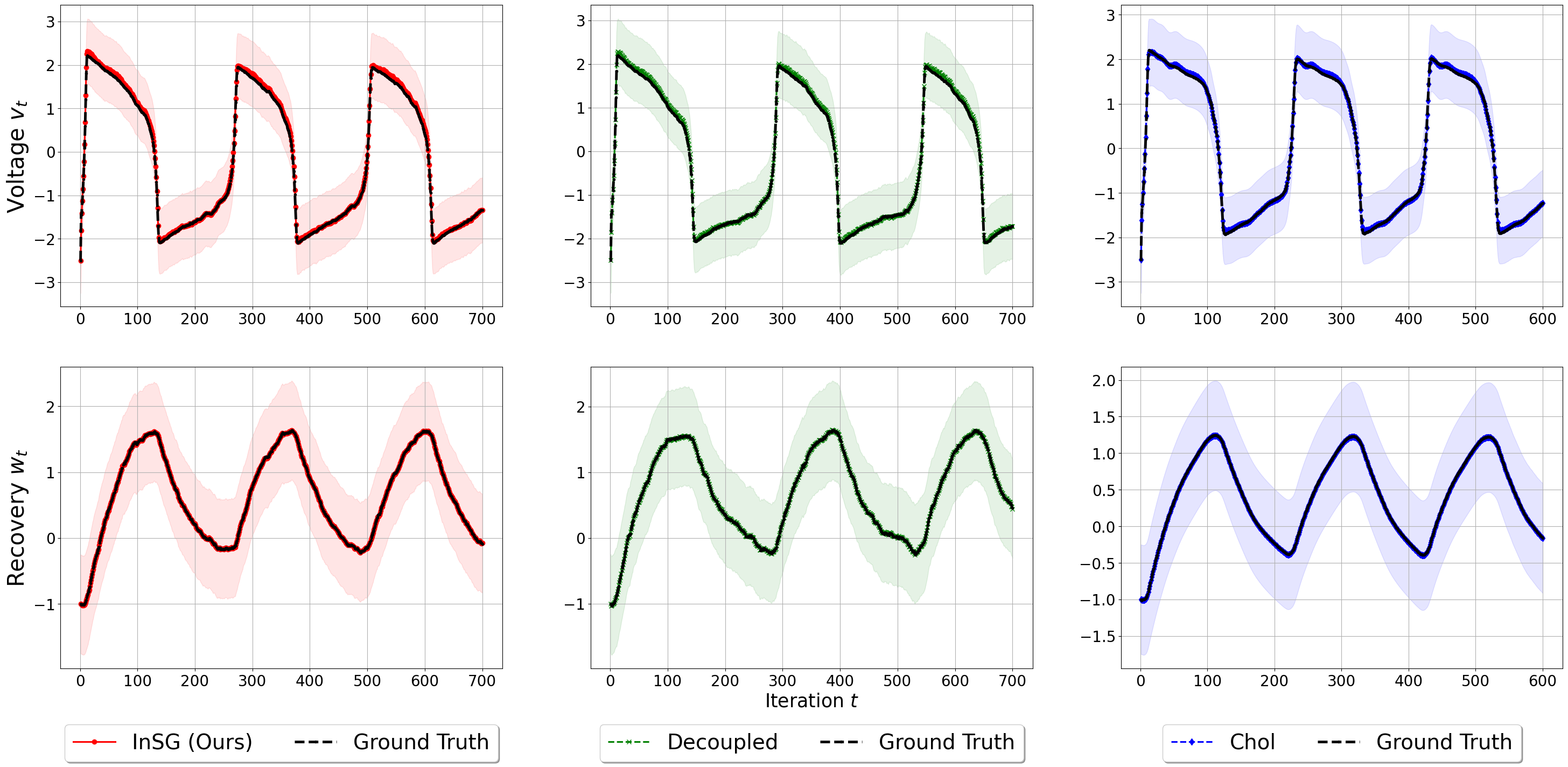}
    \caption{Trajectories of different algorithms for simulation of a stochastic FitzHugh–Nagumo neuron. Each iteration is generated over $1000$ samples. The sparse grid level of the proposed method (InSG) is set as $\eta=5$. The decoupled method is configured to use RFF priors with $2^8=64$ features and exact Matheron's update. \textit{Top:} Voltage $v_t$ trajectories generated by different algorithms. \textit{Bottom:} Recovery variable $w_t$ trajectories generated by different algorithms.}
    \label{fig:sde_trajectory}
\end{figure*}

\begin{figure}[htb]
\centering
\begin{subfigure}[b]{0.48\textwidth}
    \includegraphics[width=\linewidth]{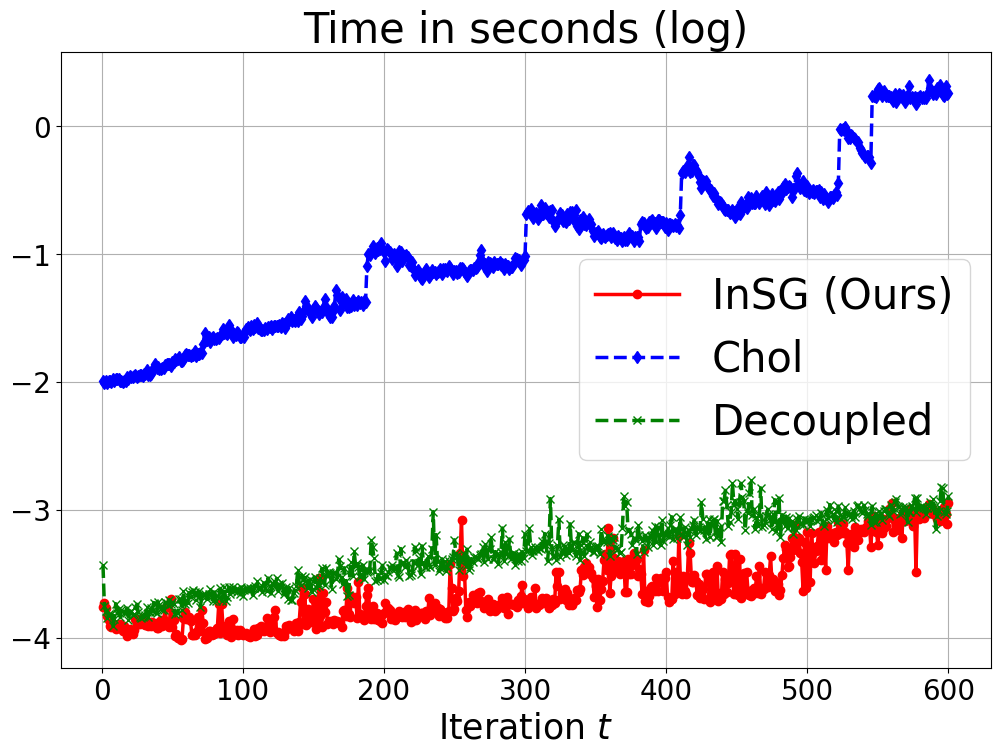}
\end{subfigure} 
\begin{subfigure}[b]{0.48\textwidth}
    \includegraphics[width=\linewidth]{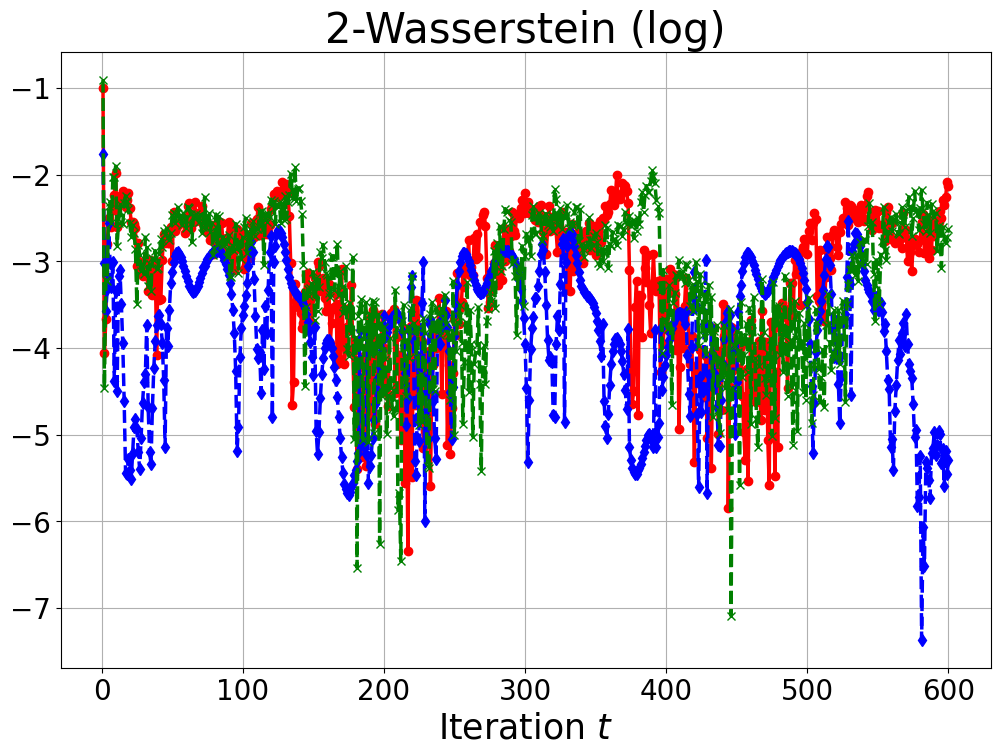}
\end{subfigure}
\caption{Time and accuracy of different algorithms for simulation of a stochastic FitzHugh–Nagumo neuron. \textit{Left:} Time cost at each iteration. \textit{Right:} Logarithm of 2-Wasserstein distance between simulations and ground truth state trajectories at each iteration.}
\label{fig:sde time_err}
\end{figure}
Our objective is to simulate the state trajectories of this dynamical system. Starting with an initial point $\bfx_0=(v_0,w_0)^T$, at iteration $t$, we draw a posterior GP sampling from the conditional distribution $p(\bfy_t\vert D_{t-1})$, where $D_{t-1}$ denotes the set of the training data $\{(\bfx_i,\bfy_i)\}_{i=1}^n$ along with the current trajectory $\{(\bfx_j,\bfy_j)\}_{j=1}^{t-1}$. In our implementation, we choose the model \cref{eq:FNmodel sde} with step size $\tau = 0.25$, parameters $\alpha=0.75$, $\beta=0.75$, $\gamma=20$, noise variance $\sigma_{\epsilon}^2 = 10^{-4}$, and initial point $\bfx_0=(-2.5,-1)^{\top}$. The training set $\{(\bfx_i,\bfy_i)\}_{i=1}^n$ was generated by evaluating \cref{eq:FNmodel sde} at $n=256$ points $\{\bfx_i\}_{i=1}^n$ $(\bfx_i \in \calX)$, which are uniformly distributed in the interval $\calX = [-2.5,2.5]\times[-1,2]$, with corresponding currents $\{a_i\}_{i=1}^{n}$ $(a_i \in \calA)$ chosen uniformly from $\calA=[0,1]$. Variations $\bfy_t$ in each step were simulated using independent Mat\'ern 3/2 GPs with $\omega=\sqrt{3}$ and inducing points approximation on the sparse grid $\calU(\eta=5,d=2)$. We generate 1000 samples at each iteration. \Cref{fig:sde_trajectory} presents the state trajectories of voltage $v_t$ and recovery variable $w_t$ versus the iteration step $t$ for each algorithm. The left plots in \Cref{fig:sde_trajectory} show that our proposed algorithm can accurately characterize the state trajectories of this dynamical system. \Cref{fig:sde time_err} illustrates the computational time required for each iteration and the 2-Wasserstein distance between the states $\bfx_t$ derived from GP-based simulations and those obtained using the Euler–Maruyama method as described in \cref{eq:sde} at each iteration $t$. These results demonstrate that our proposed algorithm not only reduces the computational time but also preserves accuracy that is comparable to the decoupled method.

% \paragraph{Regression}
% It's easy to extend our algorithm to GP regression. Suppose the inducing points are a sparse grid $\calU$ defined in \cref{eq:sgd}, then the SoR predictive mean and covariance in \eqref{eq:SoR pred distrib} can be written in the form
% \begin{subequations}
%     \begin{equation}\label{eq:SoR pred mean}
%     \bbE_{\text{SoR}}[\bff_* \mid \bfy] 
%     = \bfmu_* + \sigma^{-2}_{\epsilon} \bfK_{*,\calU} \bfSigma_{\calU}^{-1}\bfK_{\calU,\bfX}(\bfy - \bfmu_{\bfX}),
%     \end{equation}
%     \begin{equation}\label{eq:SoR pred cov}
%         \Cov_{\text{SoR}}[\bff_* \mid \bfy] 
%         = \bfK_{*,\calU} \bfSigma_{\calU}^{-1}\bfK_{\calU,*},
%     \end{equation}
% \end{subequations}
% where $\bfSigma_\calU = \big[ \bfK_{\calU,\calU} + \sigma_{\epsilon}^{-2} \bfK_{\calU,\bfX} \bfK_{\bfX,\calU} \big]$ is same as that in \cref{thm:matheron for inducing gp}. For the computation of $\Sigma_{\calU}^{-1}\bfv$ where $\bfv = \bfK_{\calU,\bfX}(\bfy - \bfmu_{\bfX})$ and $\Sigma_{\calU}^{-1} \bfK_{\calU,*}$, we implement conjugate gradient with MAS precondition to accelerate the regression. We consider UCI datasets \citep{asuncion2007uci}

\section{Conclusion}
\label{sec:conc}
In this work, we propose a scalable algorithm for GP sampling using inducing points on a sparse grid, which only requires a time complexity of $\calO \Big( (\eta+n_s) n_{sg} \Big)$ for $n_s$ sampling points, a sparse grid of size $n_{sg}$ and level $\eta$. Additionally, we show that the convergence rate for GP sampling with a sparse grid is $\calO( n_{sg}^{-\nu} (\log n_{sg})^{(\nu+2)(d-1)+d+1} )$ for a $d$-dimensional sparse grid of size $n_{sg}$ and a product Mat\'ern kernel of smoothness $\nu$. For the posterior samplings, we develop a two-level additive Schwarz preconditioner for the matrix over the sparse grid, which empirically shows rapid convergence. In the numerical study, we demonstrate that the proposed algorithm is not only efficient but also maintains accuracy comparable to other approximation methods.

%%%%%%%%%%%%%%%%%%%%%%%%%%%%%%%%%%%%%%%%%%%%%%%%%%%%%%%%%%%%
\begin{ack}
This work is supported by NSF grants DMS-2312173 and CNS-2328395.
\end{ack}

\appendix
\label{appe}

\section{Theorems}\label{sec:theos}
\subsection{Rate of convergence}
\begin{theorem}\label{theo:induce approx order}
		Let $M=L_q([0,1])$ for any $q\in[1,+\infty)$. Let $Z$ be a GP on $[0,1]$ with a Mat\'ern kernel defined in \cref{eq:Matern-1d} with smoothness $\nu$, and $\hat{Z}$ be an inducing points approximation of $Z$ under the inducing points $\{0,1/n,2/n,\ldots,1\}$ for some positive integer $n$. Then for any $p\in[1,\infty)$, as $n\rightarrow+\infty$, the order of magnitude of the approximation error is
		\[W_p(Z,\hat{Z})=\calO(n^{-\nu}).\]
	\end{theorem}
	
	\begin{proof}
		Note that $Z$ and $\hat{Z}$ already live in the same probability space, that is $(Z,\hat{Z})\in \Gamma(Z,\hat{Z})$. Therefore,
  \begin{equation}\label{eq:rate}
      W_p(Z,\hat{Z})\leq \left(\mathbb{E}\|Z-\hat{Z}\|^p_{L_q([0,1])}\right)^{1/p}.
  \end{equation}
		According to Theorem 4 in \citep{tuo2020kriging}, $\|Z-\hat{Z}\|_{L_q([0,1])}$ has the order of magnitude $O(n^{-\nu})$ and sub-Gaussian tails, which leads to the desired result.
	\end{proof}

\begin{lemma}\label{lem:rkhs sparse grid rate}
    { \color{blue} \normalfont{[Adapted from Proposition 1 and Corollary 2 in \citep{rieger2017sampling}]} } Let $\Phi:\bbR^d \rightarrow \bbR$ be a reproducing kernel of $W_{2}^{r;\otimes^d}(\bbR^d)$, where $\Phi$ is the tensor product with $\Phi(\bfx):=\prod_{j=1}^d \phi(x_j)$ and $\bfx=(x_1,\ldots,x_d)^{\top}$, $\phi:\bbR\rightarrow\bbR$ is a reproducing kernel of Sobolev space $W_2^r (\bbR^1)$ of smoothness $r>1/2$ measured in $L_2(\bbR^1)$-norm. Suppose we are given the sparse grid $\calU(\eta,d)$ defined in \cref{eq:sgd} over $I^d=[0,1]^d$ with $n=\mathrm{\normalfont{card}}(\calU(\eta,d))$ points and an unknown function $f \in W_2^{r;\otimes^d}(I^d)$. Let $s_0$ be the solution to the norm-minimal interpolant on the sparse grid $\calU(\eta,d)$ as follows:
    \begin{equation}
        \min \Vert s \Vert_{W_{2}^{r;\otimes^d}} : s \in W_{2}^{r;\otimes^d}(\bbR^d) \quad \text{with} \quad s(\xi)=f(\xi),\quad \xi \in \calU(\eta,d).
    \end{equation}
    Then, we have 
    \begin{equation}
        \Vert f-s_0 \Vert_{L_{\infty}(I^d)} \leq C n^{-r+1/2} (\log n)^{(r+3/2)(d-1)+d+1}\Vert f \Vert_{W_2^{r;\otimes^d}(\bbR^d)},
    \end{equation}
    where $C>0$ is a constant.
\end{lemma}
\begin{proof}
    See proofs of Proposition 1 and Corollary 2 in \citep{rieger2017sampling}.
\end{proof}

\begin{theorem}\label{theo:induce approx sg order}
    Let $M=L_q([0,1]^d)$ for any $q\in[1,+\infty)$, $d>1$ is the dimension of the space $M$. Let $Z$ be a GP on $[0,1]^d$ with a product Mat\'ern kernel $K(\bfx,\bfx')= \prod_{j=1}^d K_0 (x^{(j)},x'^{(j)})$ defined in \cref{eq:separable-kernel} with variance $\sigma^2=1$ and base kernel $K_0(\cdot,\cdot)$ of smoothness $\nu$, and let $\hat{Z}$ be an inducing points approximation of $Z$ under the inducing points on a sparse grid $\calU(\eta,d)$ defined in \cref{eq:sgd}. We denote the cardinality of the sparse grid $\calU(\eta,d)$ by $n$.  Then, as $n\rightarrow+\infty$, the order of magnitude of the approximation error is
    \[W_2(Z,\hat{Z}) = \calO( n^{-\nu} (\log n)^{(\nu+2)(d-1)+d+1} ) .\]
\end{theorem}
	
\begin{proof}
    If $\phi$ in \Cref{lem:rkhs sparse grid rate} is a Mat\'ern kernel in one dimension, we can infer that $r=\nu+1/2$
    by Lemma 15 in \citep{tuo2020kriging}, which states that the reproducing kernel Hilbert space (RKHS) $\calN_K(\bbR^1)$ is a Sobolev space $W_2^{\nu+1/2}(\bbR^1)$ with equivalent norms. We denote RKHS associated with product Mat\'ern kernel $K(\cdot,\cdot)$ by $\calN_K^{\otimes^d}(\bbR^d)$ and its norm by $\Vert \cdot \Vert_K$.
    
    Let $\hatf$ be the interpolation of the function $f\in\calN_{K}^{\otimes^d}(\bbR^d)$ with respect to the approximation $\hat{Z}$. Clearly, $\calN_{K}^{\otimes^d}(\bbR^d) = W_2^{\nu+1/2;\otimes^d}(\bbR^d)$. Proposition 2 in \citep{rieger2017sampling} asserts that $\hatf=s_0$ where $s_0$ is norm-minimal interpolant defined in \Cref{lem:rkhs sparse grid rate}. Therefore, for any $f \in \calN_{K}^{\otimes^d}(\bbR^d)$, we have
    \begin{equation}\label{eq:L inf sparse grid}
        \Vert f - \hatf \Vert_{L_{\infty}([0,1]^d)} \leq C n^{-\nu} (\log n)^{(\nu+2)(d-1)+d+1}\Vert f \Vert_{W_2^{\nu+1/2;\otimes^d}(\bbR^d)}.
    \end{equation}

    Additionally, Lemma 15 in \citep{tuo2020kriging} implies that
    \begin{align}\label{eq:quasi-power}
        \bbE\Vert Z-\hat{Z} \Vert_{L_q[0,1]^d}^2
        &= \sup_{\Vert f \Vert_K \leq 1} \vert f - \hatf \vert^2.
    \end{align}
    Note that $Z$ and $\hat{Z}$ already live in the same probability space, that is $(Z,\hat{Z})\in \Gamma(Z,\hat{Z})$. Therefore, by combining with \cref{eq:L inf sparse grid} and \cref{eq:quasi-power}, we can obtain the following results:
    \begin{align}\label{eq:sparse grid rate}
      W_2(Z,\hat{Z})
      &\leq \left(\mathbb{E}\|Z-\hat{Z}\|^2_{L_q[0,1]^d}\right)^{1/2}\nonumber\\
      & = \sup_{\Vert f \Vert_K \leq 1} \vert f - \hatf \vert \nonumber\\
      & \leq \sup_{\Vert f \Vert_K \leq 1} C n^{-\nu} (\log n)^{(\nu+2)(d-1)+d+1} \Vert f \Vert_{K}\nonumber\\
      & \leq C n^{-\nu} (\log n)^{(\nu+2)(d-1)+d+1}.
    \end{align}
    Thus, the order of the magnitude of the inducing points approximation on the sparse grid error is $W_2(Z,\hat{Z})=\calO(n^{-\nu} (\log n)^{(\nu+2)(d-1)+d+1})$.
\end{proof}

\subsection{Smolyak algorithm for GP sampling}
\begin{theorem}
\label{theo:sampling with sg inducing}
    Let $K(\cdot,\cdot)$ be a kernel function on $\bbR^d \times \bbR^d$, and $\calU=\calU(\eta,d)$ a sparse grid design with level $\eta$ and dimension $d$ defined in \cref{eq:sgd}. $\calU(\eta,d)=\bigcup_{\vv{t}\in\bbG(\eta)} \calU_{1,t_1} \times \calU_{2,t_2} \times \cdots \times \calU_{d,t_d}$
    where $\vv{t}=(t_1,\ldots,t_d)$, $\bbG(\eta) = \{ \vv{t}\in\bbN^d \vert \sum_{j=1}^d t_j = \eta \}$, $\eta\geq d\geq 2$, $\eta,d\in \bbN^+$. Let $f(\cdot) \sim \calGP(\bfzero, K(\cdot,\cdot))$ be a GP, then sampling from $f(\cdot)$ on a sparse grid $\calU$ can be done via the following formula:
    \begin{equation}
        f(\calU) \leftarrow \sum_{\vv{t}\in\bbP(\eta,d)} (-1)^{\eta - |\vv{t}|} \binom{d-1}{\eta - |\vv{t}|} \bigotimes_{j=1}^d \mathrm{chol}\Big( K(\calU_{j,t_j}, \calU_{j,t_j}) \Big) \bff_{\bfI_{n_{\vv{t}}}}  \delta_{x}(\calU_{\vv{t}}),
    \end{equation}
    \begin{equation}\label{eq:dirac delta}
        \delta_{x}(\calU_{\vv{t}}) := \begin{cases}
            0 \quad & \text{if $x \in \calU \setminus \calU_{\vv{t}}$}  \\
            1 \quad & \text{if $x \in \calU_{\vv{t}}$}
        \end{cases}, \qquad \calU_{\vv{t}}:=\calU_{1,t_1} \times \calU_{2,t_2} \times \cdots \times \calU_{d,t_d},
    \end{equation}
    where $\bbP(\eta,d) = \{ \vv{t} \in \bbN^d \vert \max(d,\eta-d+1) \leq \vert \vv{t} \vert \leq \eta \}$, $\mathrm{chol}\Big( K(\calU_{j,t_j}, \calU_{j,t_j}) \Big)$ is the Cholesky decomposition of $K(\calU_{j,t_j}, \calU_{j,t_j})$, $ \bff_{\bfI_{n_{\vv{t}}}} \sim \calN(\bfzero,\bfI_{n_{\vv{t}}})$ is standard multivariate normal distributed, $n_{\vv{t}}$ $=$ $\prod_{j=1}^d t_j$ is the size of the full grid design $\calU_{\vv{t}}=\calU_{1,t_1} \times \calU_{2,t_2} \times \cdots \times \calU_{d,t_d}$ associated with $\vv{t}=(t_1,\ldots,t_d)$, $\delta_{x}(\calU_{\vv{t}})$ is a Dirac measure on the set $\calU$.
\end{theorem}

\begin{proof}
    We know that sampling from $f(\calU) \sim \calN(\bfzero, K(\calU,\calU))$ can be done via Cholesky decomposition, i.e., by computing $\text{chol}\big(K(\calU,\calU) \big) \bff_{n_{sg}}$ where $\bff_{n_{sg}} \sim \calN(\bfzero, \bfI_{n_{sg}})$ is multivariate normal distributed on the sparse grid $\calU$ of size $n_{sg}=\vert \calU \vert$. Note that $\text{chol}\big(K(\calU,\calU) \big) \bff_{n_{sg}}$ can be regarded as a function $g$ evaluated on $\calU$ such that $g(\calU)=\text{chol}\big(K(\calU,\calU) \big) \bff_{n_{sg}}$. According to \cref{eq:smolyak alg}, we can obtain the following equation by applying Smolyak's algorithm $\calP(\eta,d)$ to $g$,
    \begin{equation}
    \begin{aligned}
        \calP(\eta,d) ( g ) &= \sum_{\vv{t}\in\bbP(\eta,d)} (-1)^{\eta - |\vv{t}|}  
    \binom{d-1}{\eta - |\vv{t}|} \Big( \bigotimes_{j=1}^d \calU_{j,t_j} \Big) (g)\\
    &= \sum_{\vv{t}\in\bbP(\eta,d)} (-1)^{\eta - |\vv{t}|}  
    \binom{d-1}{\eta - |\vv{t}|}
    g(\times_{j=1}^d \calU_{j,t_j})
    \delta_x (\calU_{\vv{t}})\\
    &= \sum_{\vv{t}\in\bbP(\eta,d)} (-1)^{\eta - |\vv{t}|}  
    \binom{d-1}{\eta - |\vv{t}|}
    \bigotimes_{j=1}^d
    \text{chol}\Big( K(\calU_{j,t_j}, \calU_{j,t_j}) \Big) \bff_{n_{\vv{t}}} \delta_x (\calU_{\vv{t}}),
    \end{aligned}
    \end{equation}
    where $\delta_x(\calU_{\vv{t}})$ and $\calU_{\vv{t}}$ are defined in \cref{eq:dirac delta}. Note that $g$ is defined on $\calU$, so the tensor product operation $\bigotimes_{j=1}^d \calU_{j,t_j}$ on $g$ is actually evaluating $g$ on the corresponding points, which is equivalent to applying Dirac delta measure defined on set $\times_{j=1}^d \calU_{j,t_j}$ to $g(\bigotimes_{j=1}^d \calU_{j,t_j})$.
\end{proof}

\section{Plots}
\label{sec:plots}
\begin{figure}[htb]
\centering
\includegraphics[width=\linewidth]{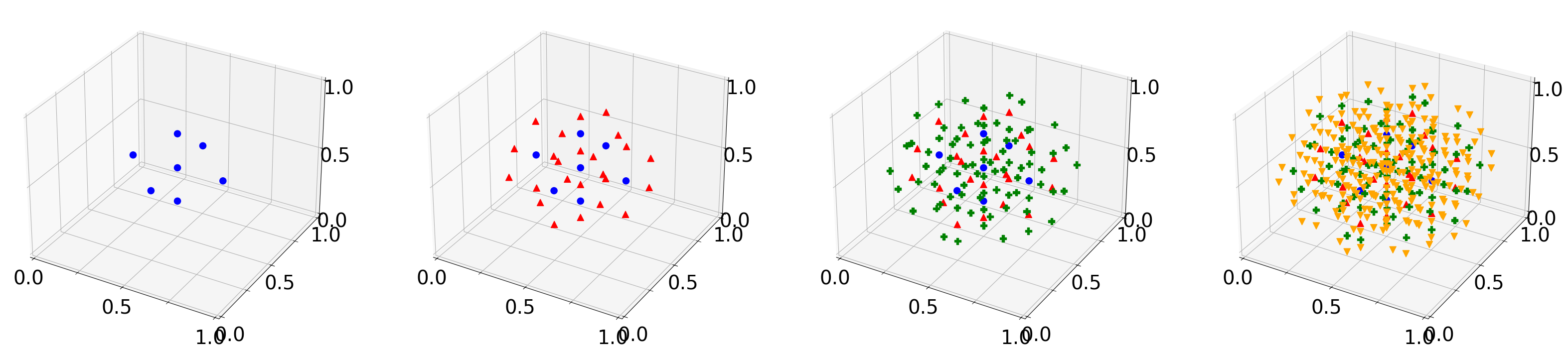}
\caption{Sparse grids $\calU(\eta,d) $over $[0,1]^d$ of level $\eta=4,5,6,7$ and dimension $d=3$.}
\label{fig:sgdesign d=3}
\end{figure}

\begin{figure}[htb]
\centering
\begin{subfigure}[b]{0.48\textwidth}
    \includegraphics[width=\linewidth]{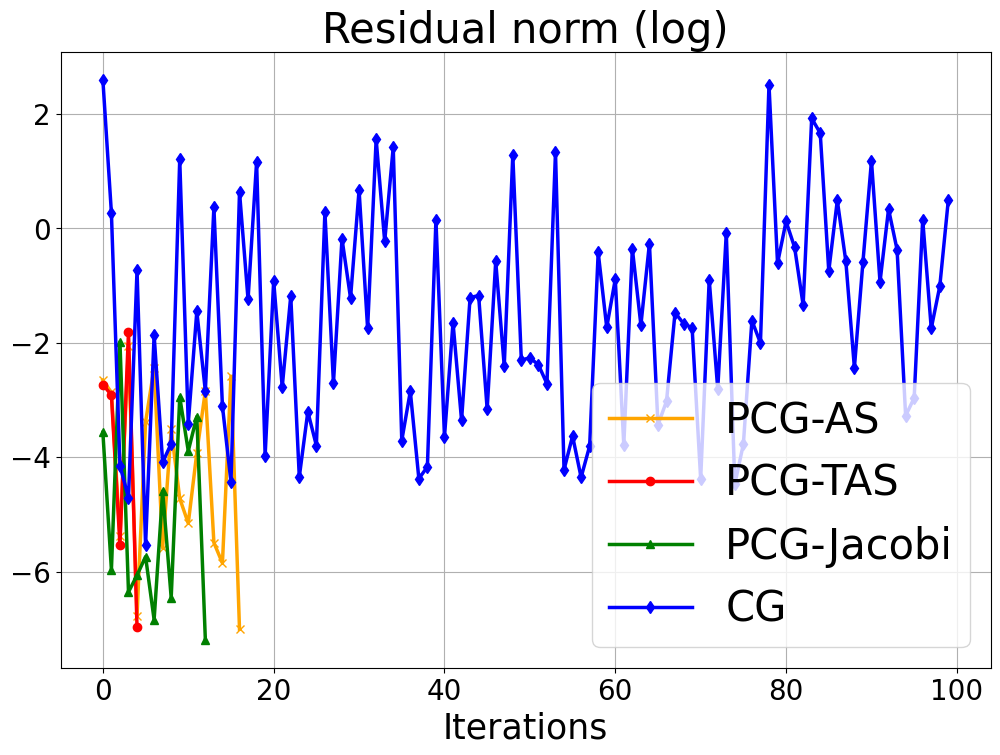}
\end{subfigure} 
\begin{subfigure}[b]{0.48\textwidth}
    \includegraphics[width=\linewidth]{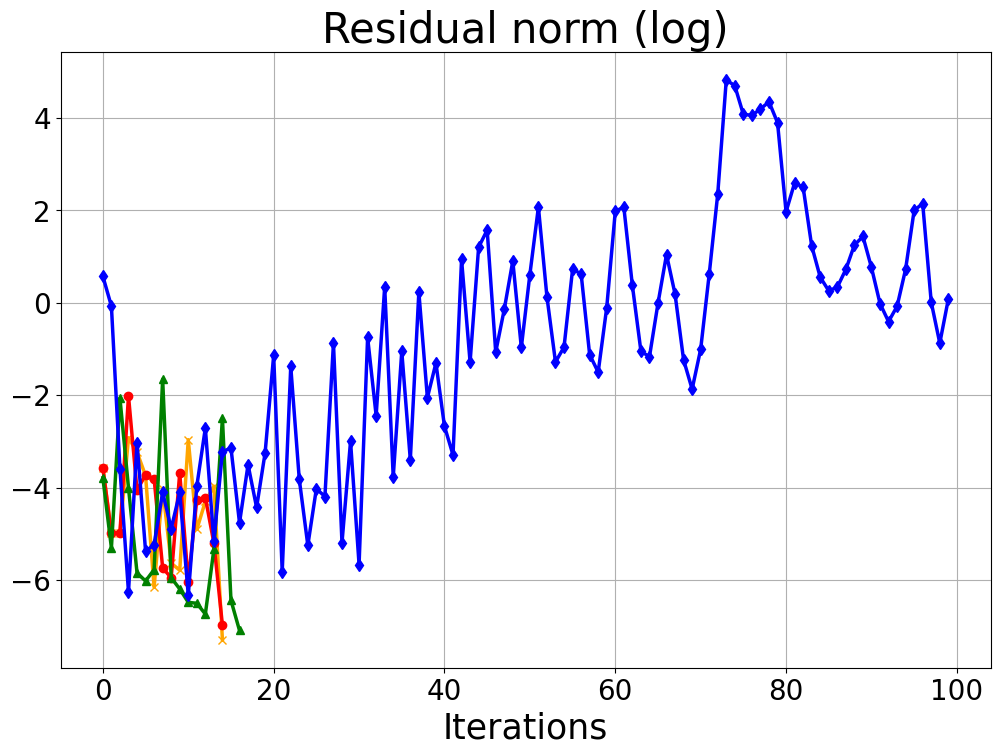}
\end{subfigure}
\caption{Residuals of PCG for $\Sigma_{\calU}^{-1} \bfv$ for tolerance $\tau=10^{-3}$ with different preconditioners where $\Sigma_{\calU}= \big[ \bfK_{\calU,\calU} + \sigma_{\epsilon}^{-2} \bfK_{\calU,\bfX} \bfK_{\bfX,\calU} \big]$, $\calU=\calU(\eta,d)$ is the sparse grid defined in \cref{eq:sgd}, vector $\bfv$ is randomly generated. \textit{Left:} The sparse grid $\calU(\eta=5,d=2)$. \textit{Right:} The sparse grid $\calU(\eta=6,d=4)$.}
\label{fig:precond blowup}
\end{figure}

\begin{figure}[htb]
\centering
\begin{subfigure}[b]{0.48\textwidth}
    \includegraphics[width=\linewidth]{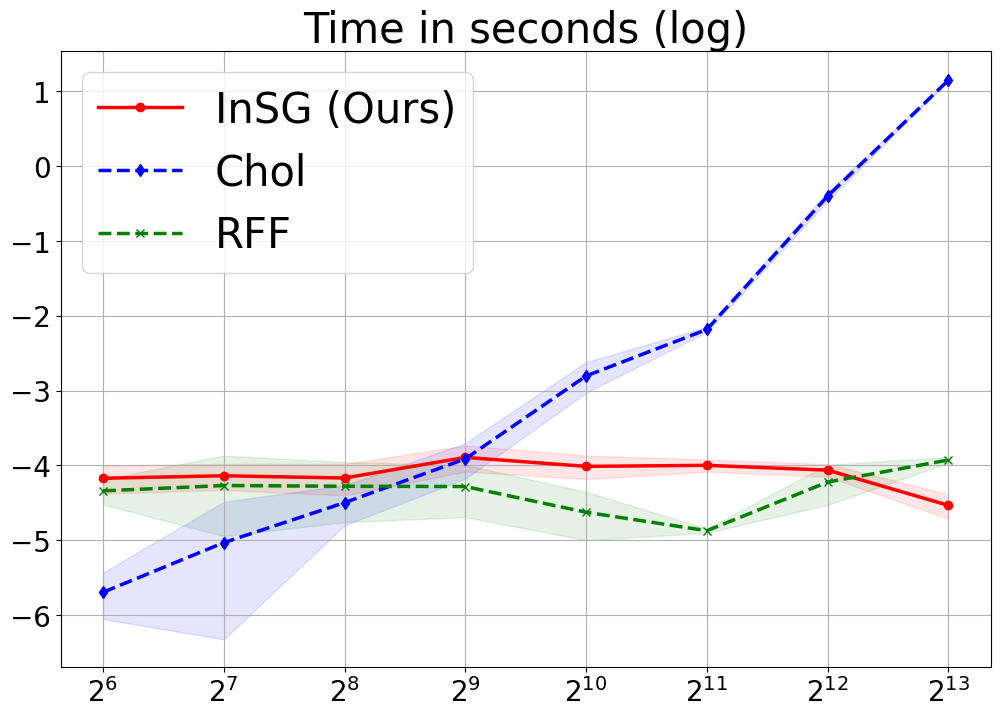}
\end{subfigure} 
\begin{subfigure}[b]{0.48\textwidth}
    \includegraphics[width=\linewidth]{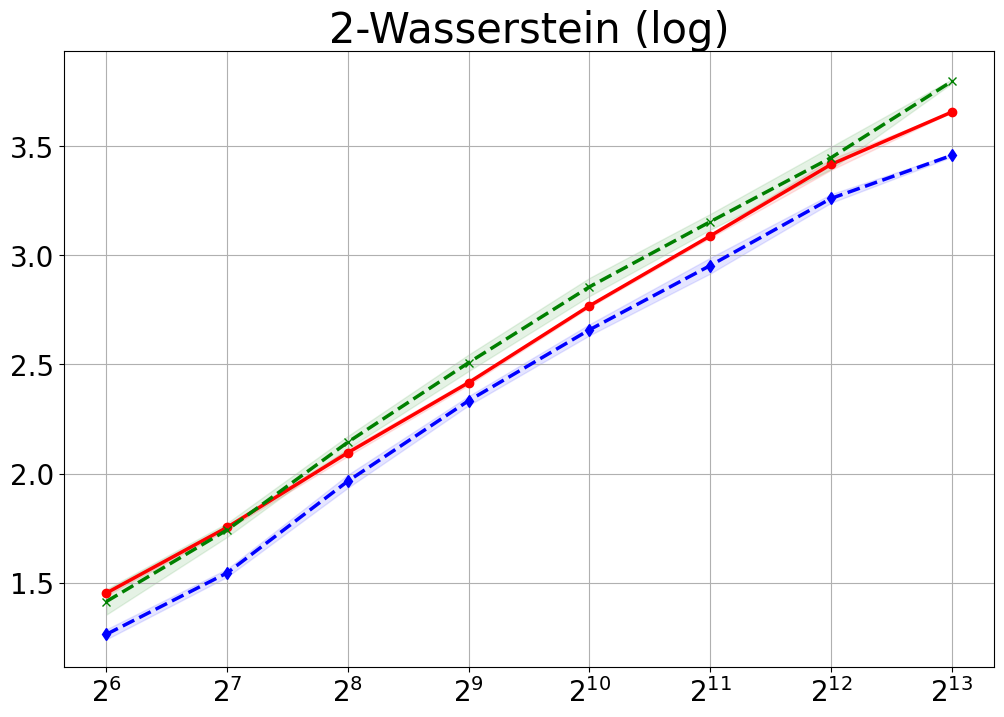}
\end{subfigure}
\caption{Time and accuracy of different algorithms for sampling from GP priors with Mat\'ern 3/2 in dimension $d=4$. $x$-axis is the number of grid points $n_s$. The sparse grid level of the proposed method is set as $\eta=6$. \textit{Left:} Logarithm of time taken to generate a draw from GP priors. \textit{Right:} Logarithm of 2-Wasserstein distances between priors and true distributions.}
\label{fig:prior d=4}
\end{figure}

\begin{figure}[htb]
\centering
\begin{subfigure}[b]{0.48\textwidth}
    \includegraphics[width=\linewidth]{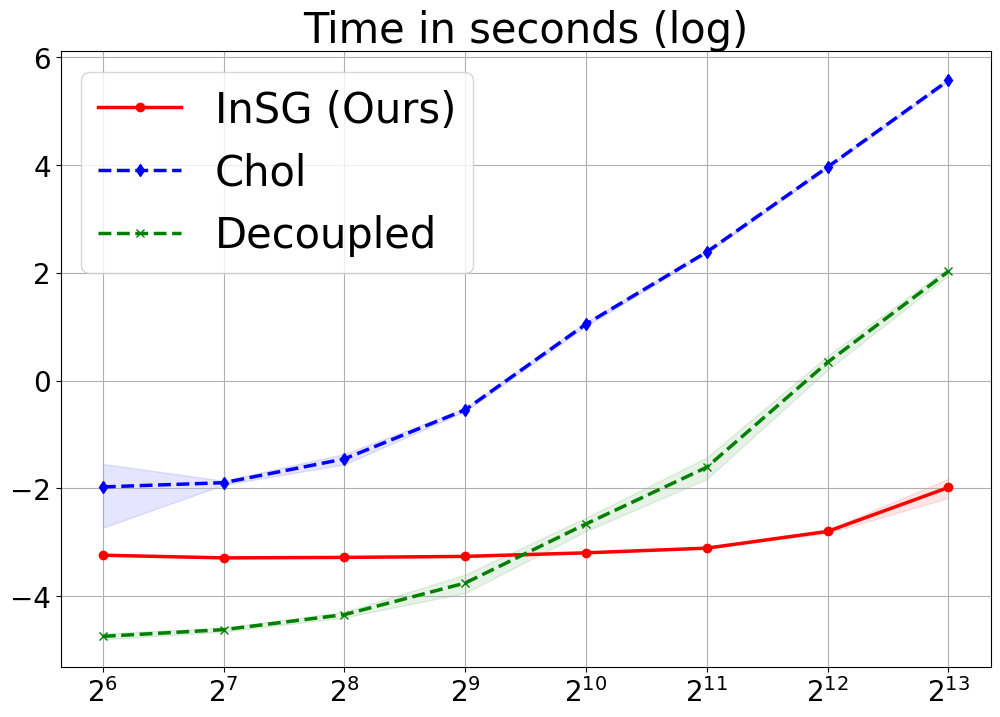}
\end{subfigure} 
\begin{subfigure}[b]{0.48\textwidth}
    \includegraphics[width=\linewidth]{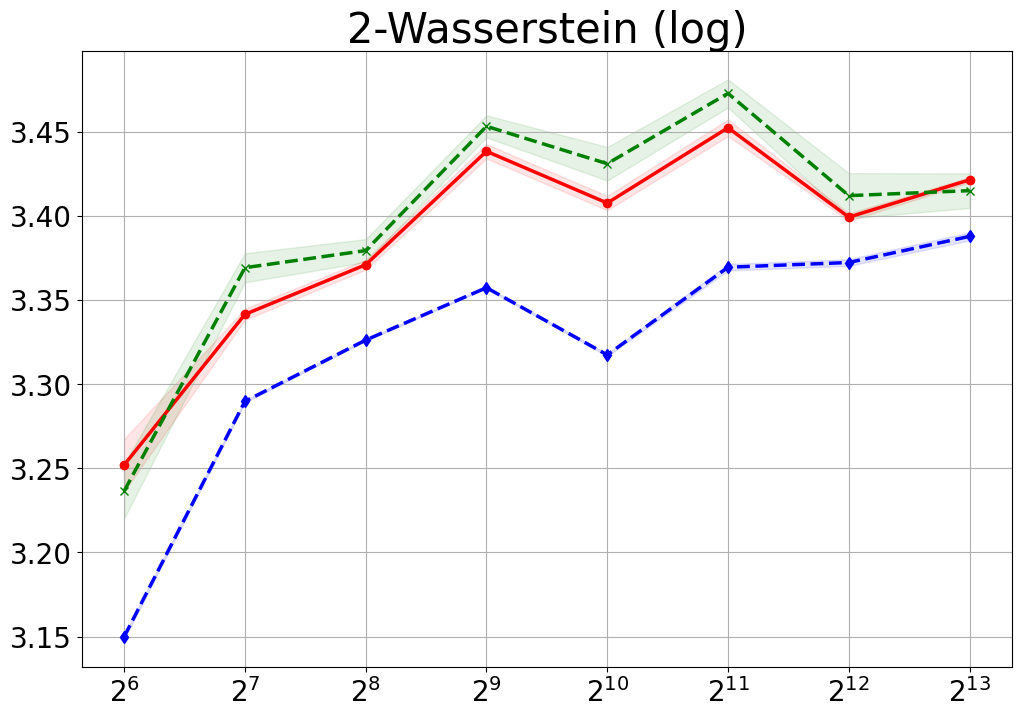}
\end{subfigure}
\caption{Time and accuracy of different algorithms for sampling from GP posteriors with Mat\'ern 3/2 and dimension $d=4$. $x$-axis is the number of observations $n$. The sparse grid level of the proposed method is set as $\eta=6$. The decoupled method is configured to use RFF priors with $2^8=64$ features and exact Matheron's update. \textit{Left:} Logarithm of time taken to generate a draw from GP posteriors over $m=1000$ points. \textit{Right:} Logarithm of 2-Wasserstein distances between posterior samplings and ground truth GP posteriors over $m=1000$ points.}
\label{fig:posterior d=4}
\end{figure}

%\section*{Acknowledgments}

%We would like to acknowledge the assistance of volunteers in putting
%together this example manuscript and supplement.

\clearpage
\bibliographystyle{apalike}
\bibliography{references}

\begin{thebibliography}{}

\bibitem[Ackley, 2012]{ackley2012connectionist}
Ackley, D. (2012).
\newblock {\em {A connectionist machine for genetic hillclimbing}}, volume~28.
\newblock Springer Science \& Business Media.

\bibitem[Al~Daas and Grigori, 2019]{al2019class}
Al~Daas, H. and Grigori, L. (2019).
\newblock {A class of efficient locally constructed preconditioners based on coarse spaces}.
\newblock {\em SIAM Journal on Matrix Analysis and Applications}, 40(1):66--91.

\bibitem[Al~Daas et~al., 2023]{al2023efficient}
Al~Daas, H., Jolivet, P., and Rees, T. (2023).
\newblock {Efficient Algebraic Two-Level Schwarz Preconditioner for Sparse Matrices}.
\newblock {\em SIAM Journal on Scientific Computing}, 45(3):A1199--A1213.

\bibitem[Banerjee et~al., 2008]{banerjee2008gaussian}
Banerjee, S., Gelfand, A.~E., Finley, A.~O., and Sang, H. (2008).
\newblock {Gaussian predictive process models for large spatial data sets}.
\newblock {\em Journal of the Royal Statistical Society Series B: Statistical Methodology}, 70(4):825--848.

\bibitem[Bishop et~al., 1995]{bishop1995neural}
Bishop, C.~M. et~al. (1995).
\newblock {\em {Neural networks for pattern recognition}}.
\newblock Oxford university press.

\bibitem[Cai and Sarkis, 1999]{cai1999restricted}
Cai, X.-C. and Sarkis, M. (1999).
\newblock {A restricted additive Schwarz preconditioner for general sparse linear systems}.
\newblock {\em Siam journal on scientific computing}, 21(2):792--797.

\bibitem[Chen et~al., 2022]{chen2022kernel}
Chen, H., Ding, L., and Tuo, R. (2022).
\newblock {Kernel packet: An Exact and Scalable Algorithm for Gaussian Process Regression with Mat{\'e}rn Correlations}.
\newblock {\em Journal of Machine Learning Research}, 23(127):1--32.

\bibitem[Chen and Tuo, 2022]{chen2022scalable}
Chen, H. and Tuo, R. (2022).
\newblock {A Scalable and Exact Gaussian Process Sampler via Kernel Packets}.

\bibitem[Cole et~al., 2021]{cole2021locally}
Cole, D.~A., Christianson, R.~B., and Gramacy, R.~B. (2021).
\newblock {Locally induced Gaussian processes for large-scale simulation experiments}.
\newblock {\em Statistics and Computing}, 31:1--21.

\bibitem[Constantinou et~al., 2017]{constantinou2017testing}
Constantinou, P., Kokoszka, P., and Reimherr, M. (2017).
\newblock {Testing separability of space-time functional processes}.
\newblock {\em Biometrika}, 104(2):425--437.

\bibitem[Cressie, 2015]{cressie2015statistics}
Cressie, N. (2015).
\newblock {\em {Statistics for spatial data}}.
\newblock John Wiley \& Sons.

\bibitem[Curriero, 2006]{curriero2006use}
Curriero, F.~C. (2006).
\newblock {On the use of non-Euclidean distance measures in geostatistics}.
\newblock {\em Mathematical Geology}, 38:907--926.

\bibitem[Cutajar et~al., 2016]{cutajar2016preconditioning}
Cutajar, K., Osborne, M., Cunningham, J., and Filippone, M. (2016).
\newblock {Preconditioning kernel matrices}.
\newblock In {\em International conference on machine learning}, pages 2529--2538. PMLR.

\bibitem[Demmel, 1997]{demmel1997applied}
Demmel, J.~W. (1997).
\newblock {\em {Applied numerical linear algebra}}.
\newblock SIAM.

\bibitem[Diggle et~al., 2003]{diggle2003introduction}
Diggle, P.~J., Ribeiro, P.~J., and Christensen, O.~F. (2003).
\newblock {An introduction to model-based geostatistics}.
\newblock {\em Spatial statistics and computational methods}, pages 43--86.

\bibitem[Ding et~al., 2021]{ding2021sparse}
Ding, L., Tuo, R., and Shahrampour, S. (2021).
\newblock A sparse expansion for deep gaussian processes.
\newblock {\em arXiv preprint arXiv:2112.05888}.

\bibitem[Dolean et~al., 2015]{dolean2015introduction}
Dolean, V., Jolivet, P., and Nataf, F. (2015).
\newblock {\em {An introduction to domain decomposition methods: algorithms, theory, and parallel implementation}}.
\newblock SIAM.

\bibitem[Dowson and Landau, 1982]{dowson1982frechet}
Dowson, D. and Landau, B. (1982).
\newblock {The Fr{\'e}chet distance between multivariate normal distributions}.
\newblock {\em Journal of multivariate analysis}, 12(3):450--455.

\bibitem[Dryja and Widlund, 1994]{dryja1994domain}
Dryja, M. and Widlund, O.~B. (1994).
\newblock {Domain decomposition algorithms with small overlap}.
\newblock {\em SIAM Journal on Scientific Computing}, 15(3):604--620.

\bibitem[Engel et~al., 2005]{engel2005reinforcement}
Engel, Y., Mannor, S., and Meir, R. (2005).
\newblock {Reinforcement learning with Gaussian processes}.
\newblock In {\em Proceedings of the 22nd international conference on Machine learning}, pages 201--208.

\bibitem[FitzHugh, 1961]{fitzhugh1961impulses}
FitzHugh, R. (1961).
\newblock {Impulses and physiological states in theoretical models of nerve membrane}.
\newblock {\em Biophysical journal}, 1(6):445--466.

\bibitem[Frazier, 2018a]{frazier2018tutorial}
Frazier, P.~I. (2018a).
\newblock {A tutorial on Bayesian optimization}.
\newblock {\em arXiv preprint arXiv:1807.02811}.

\bibitem[Frazier, 2018b]{frazier2018bayesian}
Frazier, P.~I. (2018b).
\newblock {Bayesian optimization}.
\newblock In {\em Recent advances in optimization and modeling of contemporary problems}, pages 255--278. Informs.

\bibitem[Garcke, 2013]{garcke2013sparse}
Garcke, J. (2013).
\newblock {Sparse grids in a nutshell}.
\newblock In {\em Sparse grids and applications}, pages 57--80. Springer.

\bibitem[Gardner et~al., 2018]{gardner2018gpytorch}
Gardner, J., Pleiss, G., Weinberger, K.~Q., Bindel, D., and Wilson, A.~G. (2018).
\newblock {Gpytorch: Blackbox matrix-matrix gaussian process inference with gpu acceleration}.
\newblock {\em Advances in neural information processing systems}, 31.

\bibitem[Gelbrich, 1990]{gelbrich1990formula}
Gelbrich, M. (1990).
\newblock {On a formula for the L2 Wasserstein metric between measures on Euclidean and Hilbert spaces}.
\newblock {\em Mathematische Nachrichten}, 147(1):185--203.

\bibitem[Genton, 2001]{genton2001classes}
Genton, M.~G. (2001).
\newblock {Classes of kernels for machine learning: a statistics perspective}.
\newblock {\em Journal of machine learning research}, 2(Dec):299--312.

\bibitem[Genton, 2007]{genton2007separable}
Genton, M.~G. (2007).
\newblock Separable approximations of space-time covariance matrices.
\newblock {\em Environmetrics: The official journal of the International Environmetrics Society}, 18(7):681--695.

\bibitem[Girard et~al., 2002]{girard2002gaussian}
Girard, A., Rasmussen, C., Candela, J.~Q., and Murray-Smith, R. (2002).
\newblock {Gaussian process priors with uncertain inputs application to multiple-step ahead time series forecasting}.
\newblock {\em Advances in neural information processing systems}, 15.

\bibitem[Gneiting et~al., 2006]{gneiting2006geostatistical}
Gneiting, T., Genton, M.~G., and Guttorp, P. (2006).
\newblock {Geostatistical space-time models, stationarity, separability, and full symmetry}.
\newblock {\em Monographs On Statistics and Applied Probability}, 107:151.

\bibitem[Golub and Van~Loan, 2013]{golub2013matrix}
Golub, G.~H. and Van~Loan, C.~F. (2013).
\newblock {\em {Matrix Computations}}.
\newblock JHU press.

\bibitem[Gramacy and Apley, 2015]{gramacy2015local}
Gramacy, R.~B. and Apley, D.~W. (2015).
\newblock {Local Gaussian process approximation for large computer experiments}.
\newblock {\em Journal of Computational and Graphical Statistics}, 24(2):561--578.

\bibitem[Grande et~al., 2014]{grande2014sample}
Grande, R., Walsh, T., and How, J. (2014).
\newblock {Sample efficient reinforcement learning with Gaussian processes}.
\newblock In {\em International Conference on Machine Learning}, pages 1332--1340. PMLR.

\bibitem[Griewank, 1981]{griewank1981generalized}
Griewank, A.~O. (1981).
\newblock {Generalized descent for global optimization}.
\newblock {\em Journal of optimization theory and applications}, 34(1):11--39.

\bibitem[Grigorievskiy et~al., 2017]{grigorievskiy2017parallelizable}
Grigorievskiy, A., Lawrence, N., and S{\"a}rkk{\"a}, S. (2017).
\newblock {Parallelizable sparse inverse formulation Gaussian processes (SpInGP)}.
\newblock In {\em 2017 IEEE 27th International Workshop on Machine Learning for Signal Processing (MLSP)}, pages 1--6. IEEE.

\bibitem[Hartikainen and S{\"a}rkk{\"a}, 2010]{hartikainen2010kalman}
Hartikainen, J. and S{\"a}rkk{\"a}, S. (2010).
\newblock {Kalman filtering and smoothing solutions to temporal Gaussian process regression models}.
\newblock In {\em 2010 IEEE international workshop on machine learning for signal processing}, pages 379--384. IEEE.

\bibitem[Henderson et~al., 1983]{henderson1983history}
Henderson, H.~V., Pukelsheim, F., and Searle, S.~R. (1983).
\newblock {On the history of the Kronecker product}.
\newblock {\em Linear and Multilinear Algebra}, 14(2):113--120.

\bibitem[Hensman et~al., 2018]{hensman2018variational}
Hensman, J., Durrande, N., and Solin, A. (2018).
\newblock {Variational Fourier features for Gaussian processes}.
\newblock {\em Journal of Machine Learning Research}, 18(151):1--52.

\bibitem[Hensman et~al., 2013]{hensman2013gaussian}
Hensman, J., Fusi, N., and Lawrence, N.~D. (2013).
\newblock {Gaussian processes for big data}.
\newblock {\em arXiv preprint arXiv:1309.6835}.

\bibitem[Hensman et~al., 2015]{hensman2015scalable}
Hensman, J., Matthews, A., and Ghahramani, Z. (2015).
\newblock {Scalable variational Gaussian process classification}.
\newblock In {\em Artificial Intelligence and Statistics}, pages 351--360. PMLR.

\bibitem[Hestenes et~al., 1952]{hestenes1952methods}
Hestenes, M.~R., Stiefel, E., et~al. (1952).
\newblock {Methods of conjugate gradients for solving linear systems}.
\newblock {\em Journal of research of the National Bureau of Standards}, 49(6):409--436.

\bibitem[Journel and Huijbregts, 1976]{journel1976mining}
Journel, A.~G. and Huijbregts, C.~J. (1976).
\newblock {Mining geostatistics}.

\bibitem[Katzfuss and Guinness, 2021]{katzfuss2021general}
Katzfuss, M. and Guinness, J. (2021).
\newblock {A general framework for Vecchia approximations of Gaussian processes}.
\newblock {\em Statistical Science}, 36(1):124--141.

\bibitem[Katzfuss et~al., 2022]{katzfuss2022scaled}
Katzfuss, M., Guinness, J., and Lawrence, E. (2022).
\newblock {Scaled Vecchia approximation for fast computer-model emulation}.
\newblock {\em SIAM/ASA Journal on Uncertainty Quantification}, 10(2):537--554.

\bibitem[Kennedy and O'Hagan, 2001]{kennedy2001bayesian}
Kennedy, M.~C. and O'Hagan, A. (2001).
\newblock {Bayesian calibration of computer models}.
\newblock {\em Journal of the Royal Statistical Society: Series B (Statistical Methodology)}, 63(3):425--464.

\bibitem[Kolda, 2006]{kolda2006multilinear}
Kolda, T.~G. (2006).
\newblock {Multilinear operators for higher-order decompositions.}
\newblock Technical report, Sandia National Laboratories (SNL), Albuquerque, NM, and Livermore, CA~….

\bibitem[Kuss and Rasmussen, 2003]{kuss2003gaussian}
Kuss, M. and Rasmussen, C. (2003).
\newblock {Gaussian processes in reinforcement learning}.
\newblock {\em Advances in neural information processing systems}, 16.

\bibitem[Kuss et~al., 2005]{kuss2005assessing}
Kuss, M., Rasmussen, C.~E., and Herbrich, R. (2005).
\newblock {Assessing Approximate Inference for Binary Gaussian Process Classification.}
\newblock {\em Journal of machine learning research}, 6(10).

\bibitem[Lin et~al., 2023]{lin2023sampling}
Lin, J.~A., Antor{\'a}n, J., Padhy, S., Janz, D., Hern{\'a}ndez-Lobato, J.~M., and Terenin, A. (2023).
\newblock {Sampling from Gaussian Process Posteriors using Stochastic Gradient Descent}.
\newblock {\em arXiv preprint arXiv:2306.11589}.

\bibitem[MacKay et~al., 2003]{mackay2003information}
MacKay, D.~J., Mac~Kay, D.~J., et~al. (2003).
\newblock {\em {Information theory, inference and learning algorithms}}.
\newblock Cambridge university press.

\bibitem[Maddox et~al., 2021]{maddox2021bayesian}
Maddox, W.~J., Balandat, M., Wilson, A.~G., and Bakshy, E. (2021).
\newblock {Bayesian optimization with high-dimensional outputs}.
\newblock {\em Advances in Neural Information Processing Systems}, 34:19274--19287.

\bibitem[Mallasto and Feragen, 2017]{mallasto2017learning}
Mallasto, A. and Feragen, A. (2017).
\newblock {Learning from uncertain curves: The 2-Wasserstein metric for Gaussian processes}.
\newblock {\em Advances in Neural Information Processing Systems}, 30.

\bibitem[Maruyama, 1955]{maruyama1955}
Maruyama, G. (1955).
\newblock {Continuous Markov processes and stochastic equations}.
\newblock {\em Rend. Circ. Mat. Palermo}, 4:48--90.

\bibitem[Marzouk and Najm, 2009]{marzouk2009dimensionality}
Marzouk, Y.~M. and Najm, H.~N. (2009).
\newblock {Dimensionality reduction and polynomial chaos acceleration of Bayesian inference in inverse problems}.
\newblock {\em Journal of Computational Physics}, 228(6):1862--1902.

\bibitem[Murray-Smith and Pearlmutter, 2004]{murray2004transformations}
Murray-Smith, R. and Pearlmutter, B.~A. (2004).
\newblock {Transformations of Gaussian process priors}.
\newblock In {\em International Workshop on Deterministic and Statistical Methods in Machine Learning}, pages 110--123. Springer.

\bibitem[Nagumo et~al., 1962]{nagumo1962active}
Nagumo, J., Arimoto, S., and Yoshizawa, S. (1962).
\newblock {An active pulse transmission line simulating nerve axon}.
\newblock {\em Proceedings of the IRE}, 50(10):2061--2070.

\bibitem[Neal, 2012]{neal2012bayesian}
Neal, R.~M. (2012).
\newblock {\em {Bayesian learning for neural networks}}, volume 118.
\newblock Springer Science \& Business Media.

\bibitem[Nguyen et~al., 2021]{nguyen2021gaussian}
Nguyen, V., Deisenroth, M.~P., and Osborne, M.~A. (2021).
\newblock {Gaussian Process Sampling and Optimization with Approximate Upper and Lower Bounds}.
\newblock {\em arXiv preprint arXiv:2110.12087}.

\bibitem[Nickisch and Rasmussen, 2008]{nickisch2008approximations}
Nickisch, H. and Rasmussen, C.~E. (2008).
\newblock {Approximations for binary Gaussian process classification}.
\newblock {\em Journal of Machine Learning Research}, 9(Oct):2035--2078.

\bibitem[Nickisch et~al., 2018]{nickisch2018state}
Nickisch, H., Solin, A., and Grigorevskiy, A. (2018).
\newblock {State space Gaussian processes with non-Gaussian likelihood}.
\newblock In {\em International Conference on Machine Learning}, pages 3789--3798. PMLR.

\bibitem[Nychka et~al., 2015]{nychka2015multiresolution}
Nychka, D., Bandyopadhyay, S., Hammerling, D., Lindgren, F., and Sain, S. (2015).
\newblock {A multiresolution Gaussian process model for the analysis of large spatial datasets}.
\newblock {\em Journal of Computational and Graphical Statistics}, 24(2):579--599.

\bibitem[O'Hagan, 1978]{o1978curve}
O'Hagan, A. (1978).
\newblock {Curve fitting and optimal design for prediction}.
\newblock {\em Journal of the Royal Statistical Society: Series B (Methodological)}, 40(1):1--24.

\bibitem[Okhrin et~al., 2020]{okhrin2020new}
Okhrin, Y., Schmid, W., and Semeniuk, I. (2020).
\newblock {New approaches for monitoring image data}.
\newblock {\em IEEE Transactions on Image Processing}, 30:921--933.

\bibitem[Pardo-Iguzquiza and Chica-Olmo, 2008]{pardo2008geostatistics}
Pardo-Iguzquiza, E. and Chica-Olmo, M. (2008).
\newblock {Geostatistics with the Matern semivariogram model: A library of computer programs for inference, kriging and simulation}.
\newblock {\em Computers \& Geosciences}, 34(9):1073--1079.

\bibitem[Plumlee, 2014]{plumlee2014fast}
Plumlee, M. (2014).
\newblock {Fast prediction of deterministic functions using sparse grid experimental designs}.
\newblock {\em Journal of the American Statistical Association}, 109(508):1581--1591.

\bibitem[Plumlee, 2017]{plumlee2017bayesian}
Plumlee, M. (2017).
\newblock {Bayesian calibration of inexact computer models}.
\newblock {\em Journal of the American Statistical Association}, 112(519):1274--1285.

\bibitem[Quinonero-Candela and Rasmussen, 2005]{quinonero2005unifying}
Quinonero-Candela, J. and Rasmussen, C.~E. (2005).
\newblock {A unifying view of sparse approximate Gaussian process regression}.
\newblock {\em The Journal of Machine Learning Research}, 6:1939--1959.

\bibitem[Rahimi and Recht, 2007]{rahimi2007random}
Rahimi, A. and Recht, B. (2007).
\newblock {Random features for large-scale kernel machines}.
\newblock {\em Advances in neural information processing systems}, 20.

\bibitem[Rasmussen, 2003]{rasmussen2003gaussian}
Rasmussen, C.~E. (2003).
\newblock {Gaussian processes in machine learning}.
\newblock In {\em Summer school on machine learning}, pages 63--71. Springer.

\bibitem[Rieger and Wendland, 2017]{rieger2017sampling}
Rieger, C. and Wendland, H. (2017).
\newblock {Sampling inequalities for sparse grids}.
\newblock {\em Numerische Mathematik}, 136:439--466.

\bibitem[Roberts et~al., 2013]{roberts2013gaussian}
Roberts, S., Osborne, M., Ebden, M., Reece, S., Gibson, N., and Aigrain, S. (2013).
\newblock {Gaussian processes for time-series modelling}.
\newblock {\em Philosophical Transactions of the Royal Society A: Mathematical, Physical and Engineering Sciences}, 371(1984):20110550.

\bibitem[Saad, 2003]{saad2003iterative}
Saad, Y. (2003).
\newblock {\em {Iterative methods for sparse linear systems}}.
\newblock SIAM.

\bibitem[Saat{\c{c}}i, 2012]{saatcci2012scalable}
Saat{\c{c}}i, Y. (2012).
\newblock {\em {Scalable inference for structured Gaussian process models}}.
\newblock PhD thesis, Citeseer.

\bibitem[Santner et~al., 2003]{santner2003design}
Santner, T.~J., Williams, B.~J., Notz, W.~I., and Williams, B.~J. (2003).
\newblock {\em {The design and analysis of computer experiments}}, volume~1.
\newblock Springer.

\bibitem[Shewchuk et~al., 1994]{shewchuk1994introduction}
Shewchuk, J.~R. et~al. (1994).
\newblock {An introduction to the conjugate gradient method without the agonizing pain}.

\bibitem[Silverman, 1985]{silverman1985some}
Silverman, B.~W. (1985).
\newblock {Some aspects of the spline smoothing approach to non-parametric regression curve fitting}.
\newblock {\em Journal of the Royal Statistical Society: Series B (Methodological)}, 47(1):1--21.

\bibitem[Smolyak, 1963]{smolyak1963quadrature}
Smolyak, S.~A. (1963).
\newblock {Quadrature and interpolation formulas for tensor products of certain classes of functions}.
\newblock In {\em Doklady Akademii Nauk}, volume 148, pages 1042--1045. Russian Academy of Sciences.

\bibitem[Snoek et~al., 2012]{snoek2012practical}
Snoek, J., Larochelle, H., and Adams, R.~P. (2012).
\newblock {Practical bayesian optimization of machine learning algorithms}.
\newblock {\em Advances in neural information processing systems}, 25.

\bibitem[Spillane et~al., 2014]{spillane2014abstract}
Spillane, N., Dolean, V., Hauret, P., Nataf, F., Pechstein, C., and Scheichl, R. (2014).
\newblock {Abstract robust coarse spaces for systems of PDEs via generalized eigenproblems in the overlaps}.
\newblock {\em Numerische Mathematik}, 126:741--770.

\bibitem[Srinivas et~al., 2009]{srinivas2009gaussian}
Srinivas, N., Krause, A., Kakade, S.~M., and Seeger, M. (2009).
\newblock {Gaussian process optimization in the bandit setting: No regret and experimental design}.
\newblock {\em arXiv preprint arXiv:0912.3995}.

\bibitem[Stanton et~al., 2021]{stanton2021kernel}
Stanton, S., Maddox, W., Delbridge, I., and Wilson, A.~G. (2021).
\newblock {Kernel interpolation for scalable online Gaussian processes}.
\newblock In {\em International Conference on Artificial Intelligence and Statistics}, pages 3133--3141. PMLR.

\bibitem[Stein, 1999]{stein1999interpolation}
Stein, M.~L. (1999).
\newblock {\em {Interpolation of spatial data: some theory for kriging}}.
\newblock Springer Science \& Business Media.

\bibitem[Teckentrup, 2020]{teckentrup2020convergence}
Teckentrup, A.~L. (2020).
\newblock {Convergence of Gaussian process regression with estimated hyper-parameters and applications in Bayesian inverse problems}.
\newblock {\em SIAM/ASA Journal on Uncertainty Quantification}, 8(4):1310--1337.

\bibitem[Thompson, 1933]{thompson1933likelihood}
Thompson, W.~R. (1933).
\newblock {On the likelihood that one unknown probability exceeds another in view of the evidence of two samples}.
\newblock {\em Biometrika}, 25(3-4):285--294.

\bibitem[Titsias, 2009]{titsias2009variational}
Titsias, M. (2009).
\newblock {Variational learning of inducing variables in sparse Gaussian processes}.
\newblock In {\em Artificial intelligence and statistics}, pages 567--574. PMLR.

\bibitem[Toselli and Widlund, 2004]{toselli2004domain}
Toselli, A. and Widlund, O. (2004).
\newblock {\em {Domain decomposition methods-algorithms and theory}}, volume~34.
\newblock Springer Science \& Business Media.

\bibitem[Trefethen and Bau~III, 1997]{trefethen1997numerical}
Trefethen, L.~N. and Bau~III, D. (1997).
\newblock {Numerical linear algebra, vol. 50}.

\bibitem[Tuo and Wang, 2020]{tuo2020kriging}
Tuo, R. and Wang, W. (2020).
\newblock {Kriging prediction with isotropic Mat{\'e}rn correlations: Robustness and experimental designs}.
\newblock {\em The Journal of Machine Learning Research}, 21(1):7604--7641.

\bibitem[Tuo and Wu, 2015]{tuo2015efficient}
Tuo, R. and Wu, C.~J. (2015).
\newblock {Efficient calibration for imperfect computer models}.
\newblock {\em The Annals of Statistics}, 43(6):2331 -- 2352.

\bibitem[Ullrich, 2008]{ullrich2008smolyak}
Ullrich, T. (2008).
\newblock {Smolyak's algorithm, sampling on sparse grids and Sobolev spaces of dominating mixed smoothness}.
\newblock {\em East Journal of Approximations}, 14(1):1.

\bibitem[Van~der Vorst, 2003]{van2003iterative}
Van~der Vorst, H.~A. (2003).
\newblock {\em {Iterative Krylov methods for large linear systems}}.
\newblock Number~13. Cambridge University Press.

\bibitem[Wang et~al., 2019]{wang2019exact}
Wang, K., Pleiss, G., Gardner, J., Tyree, S., Weinberger, K.~Q., and Wilson, A.~G. (2019).
\newblock {Exact Gaussian processes on a million data points}.
\newblock {\em Advances in neural information processing systems}, 32.

\bibitem[Wang et~al., 2020]{wang2020prediction}
Wang, W., Tuo, R., and Jeff~Wu, C. (2020).
\newblock {On prediction properties of kriging: Uniform error bounds and robustness}.
\newblock {\em Journal of the American Statistical Association}, 115(530):920--930.

\bibitem[Wenger et~al., 2022a]{wenger2022preconditioning}
Wenger, J., Pleiss, G., Hennig, P., Cunningham, J., and Gardner, J. (2022a).
\newblock {Preconditioning for scalable Gaussian process hyperparameter optimization}.
\newblock In {\em International Conference on Machine Learning}, pages 23751--23780. PMLR.

\bibitem[Wenger et~al., 2022b]{wenger2022posterior}
Wenger, J., Pleiss, G., Pf{\"o}rtner, M., Hennig, P., and Cunningham, J.~P. (2022b).
\newblock {Posterior and computational uncertainty in gaussian processes}.
\newblock {\em Advances in Neural Information Processing Systems}, 35:10876--10890.

\bibitem[Williams and Seeger, 2000]{williams2000using}
Williams, C. and Seeger, M. (2000).
\newblock {Using the Nystr{\"o}m method to speed up kernel machines}.
\newblock {\em Advances in neural information processing systems}, 13.

\bibitem[Williams and Rasmussen, 2006]{williams2006gaussian}
Williams, C.~K. and Rasmussen, C.~E. (2006).
\newblock {\em {Gaussian processes for machine learning}}, volume~2.
\newblock MIT press Cambridge, MA.

\bibitem[Wilson and Nickisch, 2015]{wilson2015kernel}
Wilson, A. and Nickisch, H. (2015).
\newblock {Kernel interpolation for scalable structured Gaussian processes (KISS-GP)}.
\newblock In {\em International conference on machine learning}, pages 1775--1784. PMLR.

\bibitem[Wilson et~al., 2020]{wilson2020efficiently}
Wilson, J., Borovitskiy, V., Terenin, A., Mostowsky, P., and Deisenroth, M. (2020).
\newblock {Efficiently sampling functions from Gaussian process posteriors}.
\newblock In {\em International Conference on Machine Learning}, pages 10292--10302. PMLR.

\bibitem[Wilson et~al., 2021]{wilson2021pathwise}
Wilson, J.~T., Borovitskiy, V., Terenin, A., Mostowsky, P., and Deisenroth, M.~P. (2021).
\newblock {Pathwise Conditioning of Gaussian Processes}.
\newblock {\em J. Mach. Learn. Res.}, 22:105--1.

\bibitem[Xie and Xu, 2021]{xie2021bayesian}
Xie, F. and Xu, Y. (2021).
\newblock {Bayesian projected calibration of computer models}.
\newblock {\em Journal of the American Statistical Association}, 116(536):1965--1982.

\bibitem[Yadav et~al., 2022]{yadav2022kernel}
Yadav, M., Sheldon, D.~R., and Musco, C. (2022).
\newblock {Kernel Interpolation with Sparse Grids}.
\newblock {\em Advances in Neural Information Processing Systems}, 35:22883--22894.

\bibitem[Zafari et~al., 2020]{zafari2020resolving}
Zafari, S., Murashkina, M., Eerola, T., Sampo, J., K{\"a}lvi{\"a}inen, H., and Haario, H. (2020).
\newblock {Resolving overlapping convex objects in silhouette images by concavity analysis and Gaussian process}.
\newblock {\em Journal of Visual Communication and Image Representation}, 73:102962.

\end{thebibliography}

\end{document}